\documentclass[reprint,superscriptaddress,amsmath,amssymb,aps,pra,floatfix]{revtex4-2}
\usepackage{graphicx} 
\usepackage{xcolor}
\usepackage{etoolbox}
\usepackage{physics}
\usepackage{bbm}
\usepackage{amsthm}
\usepackage{hyperref}
\hypersetup{
    colorlinks,
    citecolor=blue,
    filecolor=blue,
    linkcolor=blue,
    urlcolor=blue
}

\makeatletter
\def\maketitle{
\@author@finish
\title@column\titleblock@produce
\suppressfloats[t]}
\makeatother

\newtheorem{theorem}{Theorem}

\newtheorem*{conjecture}{Conjecture}

\DeclareMathOperator{\sgn}{sgn}
\DeclareMathOperator{\si}{si}
\DeclareMathOperator{\atantwo}{atan2}

\allowdisplaybreaks

\begin{document}

\phantomsection\addcontentsline{toc}{part}{All three-angle variants of Tsirelson's precession protocol, and improved bounds for wedge integrals of Wigner functions}

\title{All three-angle variants of Tsirelson's precession protocol,\texorpdfstring{\\}{ }and improved bounds for wedge integrals of Wigner functions}

\author{Lin Htoo Zaw}
\email[Electronic mail: ]{htoo@zaw.li}
\affiliation{Centre for Quantum Technologies, National University of Singapore, 3 Science Drive 2, Singapore 117543}

\author{Valerio Scarani}
\affiliation{Centre for Quantum Technologies, National University of Singapore, 3 Science Drive 2, Singapore 117543}
\affiliation{Department of Physics, National University of Singapore, 2 Science Drive 3, Singapore 117542}

\begin{abstract}
Tsirelson's precession protocol is a nonclassicality witness that can be defined for both discrete and continuous variable systems. Its original version involves measuring a precessing observable, like the quadrature of a harmonic oscillator or a component of angular momentum, along three equally-spaced angles. In this work, we characterise all three-angle variants of this protocol. For continuous variables, we show that the maximum score $\mathbf{P}_3^\infty$ achievable by the quantum harmonic oscillator is the same for all such generalised protocols. We also derive markedly tighter bounds for $\mathbf{P}_3^\infty$, both rigorous and conjectured, which translate into improved bounds on the amount of negativity a Wigner function can have in certain wedge-shaped regions of phase space. For discrete variables, we show that changing the angles significantly improves the score for most spin systems. Like the original protocol, these generalised variants can detect non-Gaussian and genuine multipartite entanglement when applied on composite systems. Overall, this work broadens the scope of Tsirelson's original protocol, making it capable to detect the nonclassicality and entanglement of many more states.
\end{abstract}

\maketitle

\section{Introduction}
For single-particle systems governed by Hamiltonians up to quadratic order, like a particle in free space or a harmonic oscillator, Ehrenfest's theorem tells us that the dynamical equations of the mean value of observables are exactly the same for both quantum and classical theories \cite{Ehrenfest}. Counter-intuitively, Tsirelson introduced a protocol for systems with uniformly-precessing coordinates---i.e. quadratures of a harmonic oscillator or coordinates undergoing spatial rotations---that can positively distinguish quantum systems from classical ones \cite{tsirelson-og}. He showed that if a uniformly-precessing coordinate of a classical system is measured at a randomly-chosen $\theta \in \{0,2\pi/3,4\pi/3\}$, the probability $P_3$ that the coordinate is positive is constrained by the inequality $P_3 \leq \mathbf{P}_3^c = 2/3$. The nonclassicality of a quantum system can therefore be certified by preparing a suitable state of the system that violates this inequality. This protocol depends on the assumption that both the dynamics and coordinate measurements of the system are well-characterised, but does not require the sequential measurements needed for Leggett--Garg inequalities \cite{LG-inequalities}, nor the simultaneous measurements needed for tests of noncontextuality \cite{noncontextuality}.

For the quantum harmonic oscillator, violation of the classical bound is related to the Wigner function, a quasiprobability distribution that has all the same properties as a classical joint distribution defined over phase space, except that it is allowed to be negative \cite{WignerNegativityVolume}. More precisely, the amount of violation quantifies how negative the Wigner function is over integrals of wedge-shaped phase space regions \cite{tsirelson-spin}. Using this relation, the maximum score $\mathbf{P}_3^\infty$ of the quantum harmonic oscillator was upper-bounded by Tsirelson to be $\mathbf{P}_3^\infty < 1$ \cite{tsirelson-og}, and upper-bounded using Werner's bounds of single-wedge Wigner integrals \cite{wedgeIntegralWerner} to be $\mathbf{P}_3^\infty \leq 0.822\,607$ \cite{tsirelson-spin}.

Tsirelson's precession protocol has also been extended to other oscillating systems, like spin angular momentum \cite{tsirelson-spin} and general theories \cite{tsirelson-general}, and has recently been experimentally verified in the nuclear spin of antimony \cite{tsirelson-experiment}. For continuous variables, the protocol has also been shown to be useful for detecting non-Gaussian entanglement---entanglement present in states whose Wigner functions are not Gaussian functions---of coupled harmonic oscillators using only quadrature measurements \cite{tsirelson-harmonic-entanglement}. For discrete variables, it detects genuine multipartite entanglement---entanglement that cannot be simulated by probabilistic mixtures of unentangled states---of spin ensembles using only measurements of total angular momentum \cite{tsirelson-spin-entanglement}.

Some generalisations of the original inequality have also been introduced. There, instead of equally-spaced angles, the uniformly-precessing coordinates are measured at three unequally-spaced angles, where the spacing between the angles are chosen using some other criteria \cite{tsirelson-anharmonic,tsirelson-inequalities}. Although Tsirelson did not study these generalisations, he alluded to them in a concluding remark that his upper bound $\mathbf{P}_3^\infty < 1$ holds even when the probing angles were not multiples of $2\pi/3$ \cite{tsirelson-og}.

In this work, we characterise all three-angle variants of Tsirelson's precession protocol. In ``Three-angle Precession Protocols and the Classical Inequality'', we define the protocols, discuss their symmetries, and show that the variants introduced in Refs.~\cite{tsirelson-anharmonic,tsirelson-inequalities} belong to this family. We first study the protocols on the quantum harmonic oscillator in ``Quantum Harmonic Oscillator'', where we show that all the variants achieve the same maximum score $\mathbf{P}_3^\infty$, albeit for different states. We also prove the rigorous bounds $0.709\,364 \leq \mathbf{P}_3^\infty \leq 0.730\,822$, much tighter than the previously known ones. These in turn lead to improved bounds for triple-wedge integrals over Wigner functions, tightening the fundamental limits on the amount of Wigner negativity that can be present in certain sectors of phase space. In ``Spin Angular Momentum'', we study the protocols on spin angular momenta. The maximal violation for a given spin exhibit a triangular symmetry in the parameter space, leading to an improved score for many values of the spin and removing some awkward features observed when sticking to the original protocol \cite{tsirelson-spin}. Extrapolating the limiting behaviour of the spin score, we also give the conjectured bounds $0.709\,364 \leq \mathbf{P}_3^\infty \leq 0.709\,511$ on the score of the harmonic oscillator.

When applied to composite systems, the original precession protocol was proven to be a witness of non-Gaussian entanglement for continuous variables \cite{tsirelson-harmonic-entanglement}, and of genuine multipartite entanglement for spin systems \cite{tsirelson-spin-entanglement}. In ``Implications on Detecting Entanglement'', we prove that this holds also for the generalised variants. Furthermore, we show that the separable bounds for the coupled harmonic oscillators are the same as those that were previously found with the original protocol, and that there are entangled states for both the oscillators and spin ensembles that can only be detected by a generalised protocol that is not the original one. Although we focus on the theoretical study of the protocol in this work, we consider the impact of limited measurement precision in ``Impact of Limited Measurement Precision'', which might be encountered in experimental implementations of the protocol. Finally, we also briefly touch upon the precession protocol with more than three angles in ``Implications and Outlook on Protocols with More Angles''. There, we identify the aspects of our results that are also applicable in the larger family of protocols, and provide heuristic strategies that simplify the task of finding maximally-violating states.

\section{Results}
\subsection{\label{sec2}Three-angle Precession Protocols and the Classical Inequality}
Tsirelson's precession protocol and its generalisations involve uniformly-precessing observables. A pair of observables $A_x(\theta)$ and $A_y(\theta)$ with a parametric dependence on $\theta$ are uniformly-precessing if they satisfy
\begin{equation}\label{eq:precession-condition}
    \pmqty{
        A_x(\theta) \\
        A_y(\theta)
    } = \pmqty{
        \cos\theta & \sin\theta \\
        -\sin\theta & \cos\theta
    }\pmqty{
        A_x(0)\\
        A_y(0)
    }.
\end{equation}
Examples include the position $A_x(\theta) \propto X(t=\theta T/2\pi)$ and momentum $A_y(\theta) \propto  P(t=\theta T/2\pi)$ of a harmonic oscillator with period $T$, and the angular momentum $A_x(\theta) = (\cos\theta,\sin\theta,0)\cdot\vec{J}$ of a particle along the spatial direction specified by the unit vector $(\cos\theta,\sin\theta,0)$, where $\vec{J}=(J_x,J_y,J_z)$ is the angular momentum vector, with $A_y(\theta)$ similarly defined.

Each member of the family of three-angle variants of Tsirelson's precession protocol is defined for a choice of three probing angles $\{\theta_k\}_{k=0}^2$. Up to an arbitrary offset $\theta_k \to \theta_k - \theta_0$, the probing angles can be labelled $0 =: \theta_0 \leq \theta_1 \leq \theta_2 < 2\pi$. Without any loss of generality, each protocol is therefore fully specified by the vector $\vec{\theta} := (\theta_1,\theta_2)$.

The protocol for a given $\vec{\theta}$ is applied to a pair of uniformly-precessing observables $A_x(\theta)$ and $A_y(\theta)$, and involves many independent rounds. In each round,
\begin{enumerate}
    \item The system is prepared in some state. The stability of the preparation is required to achieve a high score, but the validity of the protocol does not require the assumption that the state is the same in each round.
    \item $k \in \{0,1,2\}$ is chosen with a procedure that is uncorrelated from the preparation. The sampling distribution does not matter in the limit of large samples, as long as the average score of the protocol can be estimated to the desired precision.
    \item The coordinate $A_x(\theta_k)$ is measured: the score is $1$ if $A_x(\theta_k) > 0$, $1/2$ if $A_x(\theta_k) = 0$, or $0$ if $A_x(\theta_k) < 0$. For now, take the measurement outcomes of $A_x(\theta_k)$ to be precise enough for this score assignment to be umambigious. The impact of limited measurement precision is discussed in ``Impact of Limited Measurement Precision''.
\end{enumerate}
The round ends when the system is measured, so it does not matter if the measurement is performed destructively or not. After many rounds, we calculate the average score $P(\vec{\theta})$ of the protocol
\begin{equation}
\begin{aligned}
    P(\vec{\theta}) &:= \frac{1}{3}\sum_{k=0}^{2}\Bqty{
        \Pr[A_x(\theta_k) > 0] + \frac{1}{2}\Pr[A_x(\theta_k)=0]
    } \\
    &= \frac{1}{3}\sum_{k=0}^{2}\ev{\Theta[A_x(\theta_k)]},
\end{aligned}
\end{equation}
where $\Theta(x)$ is the Heaviside function
\begin{equation}
    \Theta(x) = \begin{cases}
        1 & \text{if $x > 0$,} \\
        1/2 & \text{if $x = 0$,}\\
        0 & \text{otherwise.}
    \end{cases}
\end{equation}
The original precession protocol introduced by Tsirelson considered $\{\theta_k\}_{k=0}^2 = \{2\pi k/3\}_{k=0}^2$. Later generalisations of the protocol include the case $\{\theta_k\}_{k=0}^2 = \{\theta_1 k\}_{k=0}^2$ with $\pi/2 \leq \theta_1 \leq \pi$ \cite{tsirelson-anharmonic} and the Type I inequality with $\{\theta_k\}_{k=0}^2 = \{4\pi k/5\}_{k=0}^2$ up to an offset of the probing angle \cite{tsirelson-inequalities}. Meanwhile, Ref.~\cite{tsirelson-inequalities} also introduced the Type II inequality, where they considered the quantity $P_{\operatorname{II}} = \sum_{k=0}^2 (-1)^k \langle \Theta[A_x(2\pi k/5)] \rangle$, again defined up to an offset of the probing angles. With the relation $\Theta(x) + \Theta(-x) = 1$ and $-\exp(i\theta) = \exp[i(\theta+\pi)]$, we can rewrite $-\langle\Theta[A_x(2\pi/5)]\rangle
= \langle\Theta[A_x(7\pi/5)]\rangle - 1$, and so $P_{\operatorname{II}} = 3 P((4\pi/5,7\pi/5))-1$. Hence, previously studied variants of Tsirelson's precession protocol are all members of this family of three-angle protocols.

In classical theory, the possible measurement outcomes of the observables $A_x$ and $A_y$ are specified by a joint probability density $f(A_x,A_y) \geq 0$ normalised as $\iint_{\mathbb{R}^2}\dd{A_x}\dd{A_y} f(A_x,A_y) = 1$, such that the expectation value of any function $g(A_x,A_y)$ of $A_x$ and $A_y$ is given by $\ev{g(A_x,A_y)} = \iint_{\mathbb{R}^2}\dd{A_x}\dd{A_y} f(A_x,A_y) g(A_x,A_y)$. Hence, the classical score is constrained by the inequality $P^c(\vec{\theta}) \leq \mathbf{P}^c(\vec{\theta})$, where
\begin{equation}
\begin{aligned}
    \mathbf{P}^c(\vec{\theta}) &:=
    \max_{\substack{
        f(A_x,A_y):f(A_x,A_y)\geq 0 \\
        \iint_{\mathbb{R}^2}\dd{A_x}\dd{A_y} f(A_x,A_y) = 1
    }}\iint_{\mathbb{R}^2}\dd{A_x}\dd{A_y}f(A_x,A_y) \\
    &\qquad{}\times{}\frac{1}{3}\sum_{k=0}^2\Theta\pqty{
        A_x \cos\theta_k  +
        A_y \sin\theta_k
    }
    \\
    &=
    \max_{(A_x',A_y')\in\mathbb{R}^2}\iint_{\mathbb{R}^2}\dd{A_x}\dd{A_y}\delta(A_x-A_x')\delta(A_y-A_y') \\
    &\qquad{}\times{}\frac{1}{3}\sum_{k=0}^2\Theta\pqty{
        A_x \cos\theta_k  +
        A_y \sin\theta_k
    }
    \\
    &= \max_{(A_x,A_y)\in\mathbb{R}^2}\frac{1}{3}\sum_{k=0}^2\Theta\pqty{
        A_x \cos\theta_k  +
        A_y \sin\theta_k
    }.
\end{aligned}
\end{equation}
Here, we have used convexity to simplify the maximisation over all probability distributions to a maximisation over just the extremal distributions $\delta(A_x-A_x')\delta(A_y-A_y')$, where $\delta(x)$ is the Dirac delta distribution. As defined, $\mathbf{P}^c(\vec{\theta})$ is the maximum score achievable in classical theory. For the original choice of angles $\vec{\theta}_3 := (2\pi/3,4\pi/3)$, Tsirelson showed that $\mathbf{P}_3^c := \mathbf{P}^c(\vec{\theta}_3) = 2/3$ \cite{tsirelson-og}. As there are quantum states that achieve the score $P_3 := P(\vec{\theta}_3) > \mathbf{P}_3^c$, nonclassicality can be certified by observing the violation of the classical inequality \cite{tsirelson-og,tsirelson-spin}.

Similarly, performing the protocol with an arbitrary choice of angles $\vec{\theta}$ and observing $P(\vec{\theta}) > \mathbf{P}^c(\vec{\theta})$ certifies the nonclassicality of the system. This requires working out the dependence of the maximum classical score $\mathbf{P}^c(\vec{\theta})$ on the probing angles $\vec{\theta}$. Intuitively, if the probing angles are spread apart far enough, $A_x(\theta)$ has to cross the line $A_x = 0$ at least once, thus $\mathbf{P}^c(\vec{\theta}) = 2/3$; but if they are too close together, $A_x(\theta)$ can remain on $A_x > 0$ plane, thus $\mathbf{P}^c(\vec{\theta}) = 1$. Indeed, we formally prove in the Supplementary Information that
\begin{equation}\label{eq:condition}
    \mathbf{P}^c(\vec{\theta}) = \begin{cases}
        2/3 = \mathbf{P}_3^c & \text{if $\forall k : \theta_{k \oplus_{3} 1} \ominus_{_{2\pi}} \theta_k \leq \pi$,} \\
        1 & \text{otherwise;}
    \end{cases}
\end{equation}
where $x\,\ooalign{\hidewidth$\pm$\hidewidth\crcr$\bigcirc$}_{\!m}\, y = (x\pm y)\bmod m$ are addition and subtraction modulo $m$. The maximum classical score $\mathbf{P}^c(\vec{\theta})$ is plotted for the full parameter space against $\vec{\theta}$ in Fig.~\ref{fig:changeOfVariable}(a).

\begin{figure}
    \centering
    \includegraphics{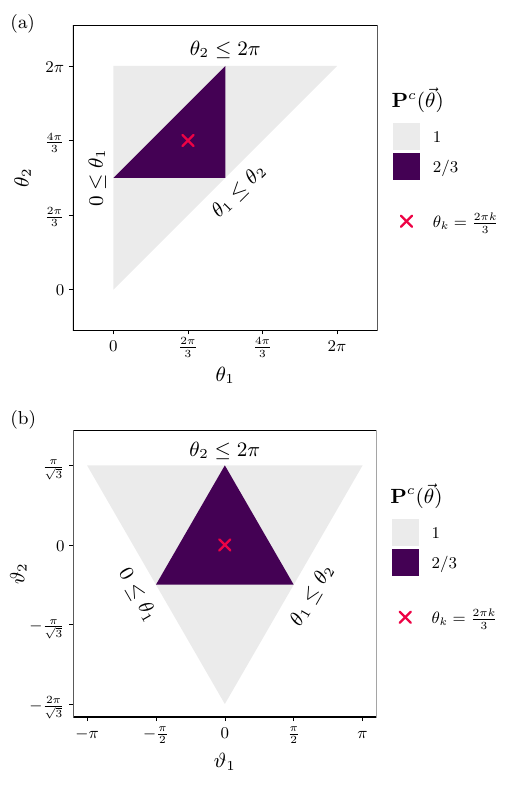}
    \caption{\label{fig:changeOfVariable}\textbf{Maximum classical score $\mathbf{P}^c(\vec{\theta})$ against probing angles.} The scores are plotted with respect to both (a) $\vec{\theta}$ and (b) $\vec{\vartheta}$. The score invariance under the equivalent probing angles given in Eq.~\eqref{eq:angles-equivalence} appear in the $\vec{\vartheta}$ coordinate as $2\pi/3$ rotations around the origin in (b).}
\end{figure}

Note also that the probing angles
\begin{equation}\label{eq:angles-equivalence}
\begin{aligned}
    \{0,\theta_1,\theta_2\} &\leftrightarrow \{- \theta_1,0,\theta_2-\theta_1\} \bmod 2\pi \\
    &\leftrightarrow\{ - \theta_2,\theta_1-\theta_2,0\} \bmod 2\pi
\end{aligned}
\end{equation}
are equivalent up to an offset $\theta_{k_0} \in \{\theta_k\}_{k=0}^2$. As such, it is more illustrative to plot the parameter space with alternate coordinates $\vec{\vartheta}$, defined as
\begin{equation}
    \vec{\vartheta} := \pmqty{
        1 & -1/2\\
        0 & \sqrt{3}/2
    }(\vec{\theta}-\vec{\theta}_3).
\end{equation}
In this coordinate, the equivalences in Eq.~\eqref{eq:angles-equivalence} become symmetry transformations
\begin{equation}\label{eq:angles-symmetry}
\vec{\vartheta} \leftrightarrow \mathcal{R}_{\frac{2\pi}{3}}^{-1}\vec{\vartheta}  \leftrightarrow \mathcal{R}_{\frac{2\pi}{3}} \vec{\vartheta},
\end{equation}
where $\mathcal{R}_{\theta} = \spmqty{\cos\theta&-\sin\theta\\\sin\theta&\cos\theta}$ is the rotation matrix in two dimensions. These alternate coordinates are again used to plot $\mathbf{P}^c(\vec{\theta})$ in Fig.~\ref{fig:changeOfVariable}(b).

For the rest of the paper, we shall use both $\vec{\theta}$ and $\vec{\vartheta}$ interchangeably. $\vec{\theta}$ is favoured when referring to the actual probing angles to be used when performing the protocol, while $\vec{\vartheta}$ is favoured when visualising quantities in the full parameter space, as it better reflects the symmetries and geometric features of the protocol.

\subsection{\label{secharm}Quantum Harmonic Oscillator}

In natural units, the quantum harmonic oscillator is governed by the Hamiltonian $H = (X^2 + P^2)/2$, where the position $X$ and momentum $P$ satisfy the canonical commutation relation $\comm{X}{P} = i$. The evolution of $X$ in time $\theta$---equivalently the $\theta$-quadrature of a bosonic field---is given by $X(\theta) = e^{i H \theta} X e^{-i H \theta} = X \cos\theta + P \sin\theta$.

As such, we can perform the precession protocol with the position of the quantum harmonic oscillator, where the maximum achievable score $\mathbf{P}^\infty(\vec{\theta})$ is
\begin{equation}\label{eq:score-operator}
    \mathbf{P}^\infty(\vec{\theta}) := \max_\rho \frac{1}{3}\sum_{k=0}^2 \tr\Bqty{\rho\Theta[X(\theta_k)]} =:
    \big\|{Q^{\infty}(\vec{\theta})}\big\|_{\infty},
\end{equation}
where $3Q^{\infty}(\vec{\theta}) := \sum_{k=0}^2 \Theta[X(\theta_k)]$ and $\|A\|_{\infty}$ is the spectral norm of $A$, which is equal to its largest eigenvalue for positive semidefinite $A \succeq 0$, as is the case for ${Q^{\infty}(\vec{\theta})}$. Meanwhile, the Heaviside function of a Hermitian operator is specified by its spectral decomposition: for any scalar function $g(a)$ and observable $A = \sum_k a_k \ketbra{a_k} + \int\dd{a}a \ketbra{a}$,
\begin{equation}\label{eq:scalar-function-operator}
    g(A) = \sum_k g(a_k) \ketbra{a_k} + \int\dd{a} g(a)\ketbra{a}.
\end{equation}

While it was known that there are states of the quantum harmonic oscillator that violate some of the generalised precession protocols, an open question remained on whether a larger violation is possible by changing the probing times \cite{tsirelson-inequalities}.

Since no violation of the classical bound is possible when $\mathbf{P}^c(\vec{\theta}) = 1$, let us focus on the subset $\triangle := \{\vec{\theta}:\mathbf{P}^c(\vec{\theta}) = 2/3\}$ of the parameter space, which corresponds to the inner purple triangle in Fig.~\ref{fig:changeOfVariable}. In the Supplementary Information, we show that the generalised protocols are related to the original one for $\vec{\theta} \in \triangle$ as
\begin{equation}\label{eq:symplectic-relation}
    Q^\infty(\vec{\theta}) = \begin{cases}
        \frac{2}{3}U(\vec{\theta})\Theta(X) U^\dag(\vec{\theta}) & \text{if $\exists k :  \theta_{k \oplus_{3} 1} \ominus_{2\pi} \theta_k = \pi$} \\
        U(\vec{\theta})Q_3^{\infty} U^\dag(\vec{\theta}) & \text{otherwise,}
    \end{cases}
\end{equation}
where $Q_3^{\infty} := Q^\infty((2\pi/3,4\pi/3))$ and $U(\vec{\theta})$ is a symplectic unitary composed of squeeze and phase shift operations. Since unitary transformations leave eigenvalues invariant, Eq.~\eqref{eq:symplectic-relation} implies that $\mathbf{P}^\infty(\vec{\theta}) = \mathbf{P}^{\infty}(\vec{\theta}_3) =: \mathbf{P}_3^\infty$ if $\vec{\theta}$ is in the interior of $\triangle$, while $\mathbf{P}^\infty(\vec{\theta}) = \mathbf{P}^c_3$ if $\vec{\theta}$ is on the boundary of $\triangle$. This dependence of $\mathbf{P}^\infty(\vec{\theta})$ on $\vec{\theta}$ is plotted in Fig.~\ref{fig:QHM}.

This means that $\mathbf{P}_3^\infty$ is the maximum achievable by the quantum harmonic oscillator for any choice of three angles, including those that define previously proposed generalisations. This maximum will however be achieved by different states: as announced, modifying the probing angles broadens the usefulness of the protocol. We also note that bounds of quantum advantages in other mechanical tasks, like the backflow constant
\begin{equation}\label{eq:backflow}
\begin{aligned}
    c_{\operatorname{B}} &:=
     \hspace{-1em}\sup_{\rho:\tr[\rho\Theta(X)]=0}\hspace{-1em}
     \tr[\rho\pqty{\Theta\bqty{X(\pi/4)}-\Theta\bqty{X(\pi/2)}}] \\
    &=
    \hspace{-1em}\sup_{\rho:\tr[\rho\Theta(X)]=0}\hspace{-1em}
    \tr[\rho\pqty{
        \Theta[X(\pi/4)] -
        \Theta[X(\pi/2)] - \Theta[X(0)]
    }] \\
    &\leq \sup_{\rho}\tr[
        \rho\pqty{
        \Theta[X(\pi/4)] +
        \Theta[X(3\pi/2)] +
        \Theta[X(\pi)]
    }] - 2  \\
    &= 3\mathbf{P}^\infty((3\pi/4,5\pi/4)) - 2,
\end{aligned}
\end{equation}
which is associated with probability backflow \cite{quantum-backflow} and quantum projectiles \cite{quantum-rockets}, is a special case of $\mathbf{P}^{\infty}(\vec{\theta})$ as labelled in Fig.~\ref{fig:QHM}.

\begin{figure}
    \centering
    \includegraphics{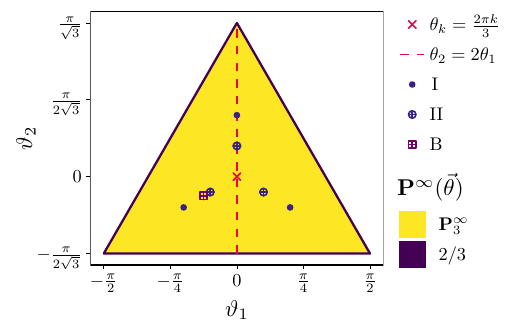}
    \caption{\label{fig:QHM} \textbf{Maximum scores $\mathbf{P}^\infty(\vec{\theta})$ for the quantum harmonic oscillator within the region where $\mathbf{P}^c < 1$ (compare to Fig.~\ref{fig:changeOfVariable}).} The score is $\mathbf{P}_3^\infty$ in the interior and $\mathbf{P}_3^c$ along the boundary. The red diagonal cross marks the original protocol with probing angles $\theta_k=2\pi k/3$ \cite{tsirelson-og}, while the red dashed lines mark $\theta_k = \theta_1 k$ \cite{tsirelson-anharmonic}. Up to a constant offset and a rescaling of $\mathbf{P}^\infty(\vec{\theta})$, the family of three-angle precession protocols also contains the Type I (solid circle) and II (circle with vertical cross) inequalities introduced in Ref.~\cite{tsirelson-inequalities}, and also provide bounds to the degree of quantum advantage in probability backflow and quantum projectiles \cite{quantum-backflow,quantum-rockets} labelled as ``B'' (square with vertical cross) in the figure.}
\end{figure}

\subsection{Relation to Negativity Volume and Wedge Integrals of Wigner Functions}
An alternate characterisation of continuous variable states is provided by the Wigner function, defined as \cite{wigner-review}
\begin{equation}
    W_\rho(x,p) := \frac{1}{2\pi i}\tr(\rho e^{i\frac{\pi}{2}[(X-x)^2+(P-p)^2)]} ),
\end{equation}
which has the property that for any scalar-valued function $g(a)$ of $a \in \mathbb{R}$,
\begin{equation}
\begin{aligned}
    &\ev{g(X\cos\theta + P\sin\theta)} \\
    &\qquad{}={} \iint_{\mathbb{R}^2}\dd{x}\dd{p} g(x\cos\theta + p\sin\theta) W_\rho(x,p),
\end{aligned}
\end{equation}
where the operator-valued function $g(A)$ of observable $A$ is specified by the action of the function on the spectral decomposition of its argument, as given in Eq.~\eqref{eq:scalar-function-operator}.

Therefore, the Wigner function acts like a probability density function of $x$ and $p$, although it is only a \emph{quasi}probability distribution as $X$ and $P$ are not jointly measurable in quantum theory. Nonetheless, it allows us to simultaneously study the classical and quantum harmonic oscillator in terms of an initial (quasi)probability distribution $F(x,p)$, where $F(x,p)$ is a joint probability density function in the classical case and $F(x,p) = W_\rho(x,p)$ in the quantum case. Then, the score $P(\vec{\theta})$ achieved by the oscillator is
\begin{equation}\label{eq:score-wigner}
\begin{aligned}
    P(\vec{\theta})
    &= \sum_{k=0}^2\frac{1}{3} \iint_{\mathbb{R}^2}\dd{x}\dd{p}\Theta\pqty{
        x\cos\theta_k +
        p\sin\theta_k
    } F(x,p) \\
    &= \mathbf{P}_3^c + \frac{1}{3}\pqty{\iint_{\Omega_+(\vec{\theta})}\dd{x}\dd{p} F(x,p)-1} \\
    &= \mathbf{P}_3^c - \frac{1}{3}\iint_{\Omega_-(\vec{\theta})}\dd{x}\dd{p} F(x,p),
\end{aligned}
\end{equation}
where $\Omega_{\pm}(\vec{\theta})$ are the phase space regions
\begin{equation}\label{eq:triple-wedge-definition}
\begin{aligned}
    \Omega_\pm(\vec{\theta}) = \big\{&(x,p) : \forall k \in\{0,1,2\} : \\
    &\qquad\theta_{k \oplus_3 1} \mp \frac{\pi}{2} \leq \atantwo(p,x) \leq \theta_k \pm \frac{\pi}{2} \big\},
\end{aligned}
\end{equation}
which are illustrated in Fig.~\ref{fig:wedge-integral}(c). Here, $\atantwo(p,x)$ is the principle value of the complex argument of $x + ip$; i.e., $-\pi \leq \atantwo(p,x) \leq \pi$ and $\sqrt{x^2+p^2} e^{i\atantwo(p,x)} = x + ip$.

\begin{figure}
    \centering
    \includegraphics[width=\columnwidth]{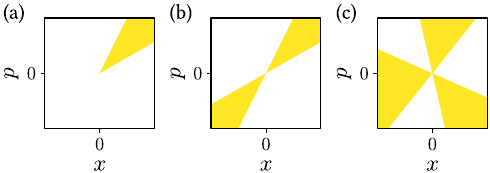}
    \caption{\label{fig:wedge-integral}\textbf{Wedge integrals of Wigner functions.} These are the (a) single $\Omega_1$, (b) double $\Omega_2$, and (c) triple $\Omega_3$ (equivalently $\Omega_\pm$ in Eq.~\eqref{eq:triple-wedge-definition}) wedges. The single wedge is a pointed sector stemming from the origin; while the double and triple wedges are, respectively, formed by two and three lines that cross the origin, and can also be seen as unions of single wedges. The exact angles do not matter, as different choices of angles are equivalent up to a symplectic transformation.}
\end{figure}

If $F(x,p)$ satisfies the properties of a joint probability density, that is if $0 \leq \iint_{\omega}\dd{x}\dd{p} F(x,p) \leq 1$ for every phase space region $\omega$, then $\abs{\frac{1}{2}-\iint_{\omega}\dd{x}\dd{p} F(x,p)} \leq \frac{1}{2}$ and Eq.~\eqref{eq:score-wigner} implies that $P(\vec{\theta}) \leq \mathbf{P}_3^c$. Therefore, the violation of the classical bound can be understood as the impossibility of assigning a classical probability distribution to $x$ and $p$. A quantification of this impossibility is the Wigner negativity volume \cite{WignerNegativityVolume}
\begin{equation}
\begin{aligned}
    \mathcal{N}_V(\rho) &:= -\iint_{\mathbb{R}^2}\dd{x}\dd{p}\Theta[-W_\rho(x,p)]W_\rho(x,p) \\
    &= \frac{1}{2}\iint_{\mathbb{R}^2}\dd{x}\dd{p}\bqty\big{\abs{W_\rho(x,p)}-W_\rho(x,p)},
\end{aligned}
\end{equation}
which is a measure of nonclassicality in the resource theories of non-Gaussianity and Wigner negativity \cite{NonGaussianResourceTheory1,NonGaussianResourceTheory2}.

The amount of violation of the precession protocol can therefore be interpreted as the lower bound for the negativity volume, since rearranging the last line of Eq.~\eqref{eq:score-wigner} gives $3P(\vec{\theta})-3\mathbf{P}_3^c \leq \mathcal{N}_V(\rho)$.

We can also rearrange the last two lines of Eq.~\eqref{eq:score-wigner} to
\begin{equation}
\begin{aligned}
    P^{\infty}(\vec{\theta}) - \frac{1}{2} &
    =
            \frac{1}{3}\pqty{
                \frac{1}{2} - \frac{1}{3}\iint_{\Omega_-(\vec{\theta})}\!\!\dd{x}\dd{p} W_\rho(x,p)
            }
    \\
    &= -\frac{1}{3}\pqty{\frac{1}{2} - \iint_{\Omega_+(\vec{\theta})}\!\!\dd{x}\dd{p} W_\rho(x,p)
    },
\end{aligned}
\end{equation}
which, together with $2\mathbf{P}_3^c - 1 = 1/3$ and $\Theta(x) + \Theta(-x) = 1 \implies \min_{F(x,p)}P(\vec{\theta}) = 1 - \max_{F(x,p)}P(\vec{\theta})$, implies
\begin{equation}\label{eq:wedge-integral-score-bounds}
    \abs{\frac{1}{2}-\iint_{\Omega_\pm(\vec{\theta})}\!\!\dd{x}\dd{p} W_\rho(x,p)} = \frac{
        \abs{P^{\infty}(\vec{\theta})-\frac{1}{2}}
    }{
        2(\mathbf{P}_3^c - \frac{1}{2})
    } \leq  \frac{
        \mathbf{P}^{\infty}_3 - \frac{1}{2}
    }{
        2(\mathbf{P}_3^c - \frac{1}{2})
    },
\end{equation}
so the maximum quantum score $\mathbf{P}_3^\infty$ also fundamentally bounds integrals of Wigner functions over certain phase space regions. In fact, integrals over $\Omega_\pm(\vec{\theta})$ take the form of wedge integrals over phase space regions: these include the (single) wedge $\Omega_1 = \{(x,p):\theta_1 \leq \atantwo(p,x) \leq \theta_2\}$ and the double wedge $\Omega_2 = \{(x,p): \theta_1 \leq \atantwo(p,x) \leq \theta_2 \lor \theta_1 \leq \atantwo(p,x)-\pi \leq \theta_2 \}$, both for $\theta_1 < \theta_2$, the latter so called because it is a union of two wedges. Both are illustrated in Fig.~\ref{fig:wedge-integral}. Bounds for the single and double wedge integrals have been worked out to be \cite{wedgeIntegralWerner,wedgeIntegralWoodBracken}
\begin{equation}
\begin{aligned}
    \sup_\rho \abs{\frac{1}{2}-\iint_{\Omega_1}\dd{x}\dd{p} W_\rho(x,p)} &\approx 0.655\,940,\\
    \sup_\rho \abs{\frac{1}{2}-\iint_{\Omega_2}\dd{x}\dd{p} W_\rho(x,p)} &\approx 0.736\,824.
\end{aligned}
\end{equation}
Analogously, $\Omega_3 = \Omega_{\pm}$ is a union of three wedges and therefore a ``triple wedge'', for which
\begin{equation}\label{3wedge}
    \sup_\rho \abs{\frac{1}{2}-\iint_{\Omega_3}\dd{x}\dd{p} W_\rho(x,p)} = \frac{
        \mathbf{P}^{\infty}_3 - \frac{1}{2}
    }{
        2(\mathbf{P}_3^c - \frac{1}{2})
    }.
\end{equation}

\subsection{\label{sec:improved-rigorous-bounds}Improved Rigorous Bounds}
Prior to this work, the best known rigorous bounds of the maximum quantum score were \cite{tsirelson-spin}
\begin{equation}\label{previous}
    0.708\,741 \leq \mathbf{P}_3^\infty \leq 0.822\,607.
\end{equation}
In particular, the old bound $0.822\,607$ was found by taking the triple wedge integral as the sum of three single wedge integrals, and using the single wedge bounds derived by \citet{wedgeIntegralWerner} to find
\begin{equation}\label{3wedgeWerner}
    \sup_\rho \abs{\frac{1}{2}-\iint_{\Omega_3}\dd{x}\dd{p} W_\rho(x,p)} \leq 0.967\,820,
\end{equation}
then using Eq.~\eqref{eq:score-wigner} to obtain the upper bound.

Now, we shall first present improved rigorous bounds of $\mathbf{P}_3^\infty$, then use them to improve the bound for the negativity of the triple wedge. To start, we define the observable
\begin{equation}
\begin{aligned}
    A_3 &:= \pqty\Big{Q_3^\infty-\pqty{1-\mathbf{P}_3^c}\mathbbm{1}}\pqty\Big{Q_3^\infty-\mathbf{P}_3^c\mathbbm{1}} \\
    &= \pqty{Q_3^\infty-\frac{1}{2}\mathbbm{1}}^2 - \pqty{\mathbf{P}_3^c-\frac{1}{2}}^2\mathbbm{1},
\end{aligned}
\end{equation}
from which lower and upper bounds of $\|A_3\|_{\infty}$ can be rearranged into lower and upper bounds of $\mathbf{P}_3^\infty$ using $\|A_3\|_{\infty} = (\mathbf{P}_3^\infty-1/2)^2-(\mathbf{P}_3^c-1/2)^2$.

The reason for working with $A_3$ instead of $Q_3^\infty$ is two-fold. First, the lower bound will be obtained by truncating the operator and finding its maximum eigenvalue, for which the sequence of lower bounds was found to converge faster with $A_3$ than with $Q_3^\infty$. Second, the upper bound will be obtained from the trace of the observable, which requires removing the values $\mathbf{P}_3^c$ and $1-\mathbf{P}_3^c$ that are in the continuous spectra of $Q_3^\infty$ \cite{tsirelson-og}.

In the Supplementary Information, we derive an expression for the matrix elements of $A_3$ in the number basis $\{\ket{n}\}_{n=0}^\infty$, where $(X^2+P^2)\ket{n} = (2n+1)\ket{n}$. $\mel{n}{A_3}{n'}$ is given in terms of the incomplete beta function and the generalised hypergeometric function, which are commonly-used special functions that can be computed with standard numerical libraries to arbitrary precision.

The lower bound $\mathbf{P}_3^{\leq}(\mathbf{n}) \leq \mathbf{P}_3^\infty$ is then obtained by truncating $A_3$ to the subspace spanned by the first $6\mathbf{n} + 1$ number states $\{\ket{n}\}_{n=0}^{6\mathbf{n}}$ and solving for its maximum eigenvalue. $\mathbf{P}_3^{\leq}(\mathbf{n})$ is plotted against $\mathbf{n}$ in Fig.~\ref{fig:lowerBound}. The largest lower bound we obtained is $\mathbf{P}_3^{\leq}(400) \geq 0.709\,364$: in fact, the sequence of lower bounds appears to saturate to this value, which would imply that $\lim_{\mathbf{n}\to\infty}\mathbf{P}_3^{\leq}(\mathbf{n}) = \mathbf{P}_3^\infty \approx 0.709\,364$. Furthermore, fitting the sequence of lower bounds to the ansatz $\mathbf{P}_3^{\leq}(\mathbf{n}) = \mathbf{P}_3^{\infty} - a_1 (\mathbf{n}+1)^{-\frac{1}{2}} - a_2 (\mathbf{n}+1)^{-\frac{3}{2}} - \mathcal{O}[(\mathbf{n}+1)^{-\frac{5}{2}}]$ gives $\mathbf{P}_3^{\infty} \approx 0.709\,364$, which corroborates the observation.

\begin{figure}
    \centering
    \includegraphics[width=\columnwidth]{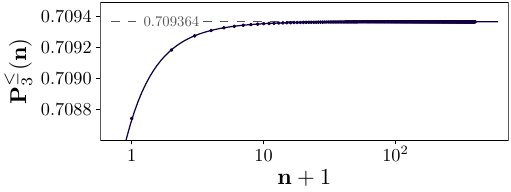}
    \caption{\label{fig:lowerBound}\textbf{Sequence of lower bounds $\mathbf{P}_3^{\leq}(\mathbf{n}) \leq \mathbf{P}_3^{\infty}$ against cutoff number $\mathbf{n}$.} The numerical fit with the ansatz $\mathbf{P}_3^{\leq}(\mathbf{n}) = \mathbf{P}_3^{\infty} - a_1 (\mathbf{n}+1)^{-\frac{1}{2}} - a_2 (\mathbf{n}+1)^{-\frac{3}{2}} - \mathcal{O}[(\mathbf{n}+1)^{-\frac{5}{2}}]$ is also plotted. The fitted parameters give $\mathbf{P}_3^{\infty}\approx  0.709\,364$, which corroborates the observation that $\lim_{\mathbf{n}\to\infty}\mathbf{P}_3^{\leq}(\mathbf{n}) = \mathbf{P}_3^{\infty}\approx 0.709\,364$.}
\end{figure}

Meanwhile, since $\Pi Q_3^\infty \Pi = \mathbbm{1} - Q_3^\infty$ and $\Pi A_3^\infty \Pi = A_3^\infty$ for the parity operator $\Pi$ that satisfies $\Pi X(\theta) \Pi = -X(\theta)$, every nonzero eigenvalue of $A_3$ is doubly degenerate. Therefore, an upper bound for $\mathbf{P}_3^\infty$ can be obtained using
\begin{equation}
\begin{aligned}
    \tr(A_3^2) = \frac{6\log 2}{(18\pi)^2} &\geq 2\|A_3\|_{\infty}^2 \\
    &= 2\bqty{\pqty{\mathbf{P}_3^\infty-\frac{1}{2}}^2-\pqty{\mathbf{P}_3^c-\frac{1}{2}}^2}^2.
\end{aligned}
\end{equation}
The analytical evaluation of the trace is given in the Supplementary Information. Relating this upper bound to $\mathbf{P}_3^\infty$ gives
\begin{equation}
    \mathbf{P}_3^\infty \leq \frac{1}{2}\pqty{1 + \frac{1}{3}\sqrt{1+\frac{2}{\pi}\sqrt{3\log2}}} \leq
    0.730\,822.
\end{equation}
In summary, we obtain rigorous lower and upper bounds \begin{align}
    0.709\,364 \leq \mathbf{P}_3^\infty \leq  0.730\,822,
\end{align} a stark improvement over Eq.~\eqref{previous}. The sequence as plotted in Fig.~\ref{fig:lowerBound} also strongly implies the tightness of the lower bound.

We finish by highlighting two consequences of the improved bounds. First, it is now rigorously proved that $\mathbf{P}_3^\infty<3/4$, which is a value that can be achieved with suitable states of spin $j=3/2$ and equally spaced angles (see Ref.~\cite{tsirelson-spin} and the section ``Spin Angular Momentum'' below). Thus, the score of the Tsirelson protocol is higher for finite-dimensional systems than for the quadratures $X(\theta)$ of continuous variable systems. Second, using the new bounds, the bound of the triple wedge integral Eq.~\eqref{3wedge} becomes
\begin{equation}
    0.628\,092 \leq \sup_\rho \abs{\frac{1}{2}-\iint_{\Omega_3}\dd{x}\dd{p} W_\rho(x,p)} \leq 0.692\,464,
\end{equation}
significantly tightening Eq.~\eqref{3wedgeWerner}.

\subsection{\label{secspin}Spin Angular Momentum}
The first generalisation of Tsirelson's original protocol extended it from phase space to real space by considering the precession of the angular momentum of a system \cite{tsirelson-spin}. A rotation of $J_x$ by an angle $\theta$ around the $-z$ axis reads
\begin{equation}
    J_x(\theta) = e^{-i\theta J_z} J_x e^{i\theta J_z} = J_x \cos\theta + J_y \sin\theta .
\end{equation}
The observed score upon performing the precession protocol on $J_x(\theta)$ is $P(\vec{\theta}) = {\tr}[\rho Q(\vec{\theta})]$, where
\begin{equation}
    Q(\vec{\theta}) = \frac{1}{3}\sum_{k=0}^2 e^{-i\theta_k J_z}  \Theta(J_x) e^{i\theta_k J_z}.
\end{equation}
In terms of the simultaneous eigenstates $\ket{j,m}$ of $|\vec{J}|^2$ and $J_z$ with eigenvalues $j(j+1)$ and $m$, respectively, the matrix elements of $Q(\vec{\theta})$ are known analytically, and it is also known that $Q(\vec{\theta}) = \bigoplus_j Q^{(j)}(\vec{\theta})$ can be decomposed into blocks with irreducible spin $j$ \cite{tsirelson-spin}.
As the score $P(\vec{\theta}) = {\tr}[\rho Q(\vec{\theta})] = \sum_j {\tr}[\rho Q^{(j)}(\vec{\theta})]$ is simply a sum of the expectation value on each $Q^{(j)}(\vec{\theta})$ block, we can restrict our analysis to the precession protocol performed for a fixed spin $j$.

For the rest of the paper, we shall denote the maximum quantum score for a fixed spin $j$ as $\mathbf{P}^{(j)}(\vec{\theta}) := \max_{\rho}{\tr}[\rho Q^{(j)}(\vec{\theta})]$, with analogous shorthands $Q^{(j)}_3 := Q^{(j)}(\vec{\theta}_3)$ and $\mathbf{P}^{(j)}_3 := \mathbf{P}^{(j)}(\vec{\theta}_3)$ for the original protocol.

Unlike the harmonic oscillator case, no continuous symplectic transformations exist for the angular momentum operators, so the choice of probing angles affects $\mathbf{P}^{(j)}(\vec{\theta})$, even in the interior region $\triangle$ where $\mathbf{P}^c(\vec{\theta}) = 2/3$. As such, it is of interest to find the choices of $\vec{\theta}$ that demonstrate large violations of the classical bound.

The maximum scores $\mathbf{P}^{(j)}(\vec{\theta})$, calculated with standard numerical tools to diagonalise the finite dimensional matrices $Q^{(j)}(\vec{\theta})$, are plotted against $\vartheta_1$ and $\vartheta_2$ as heatmaps for some select values of half integer spins in Fig.~\ref{fig:heatmap-half-spins} and integer spins in Fig.~\ref{fig:heatmap-integer-spins}.

\begin{figure*}
    \centering
    \includegraphics[width=\textwidth]{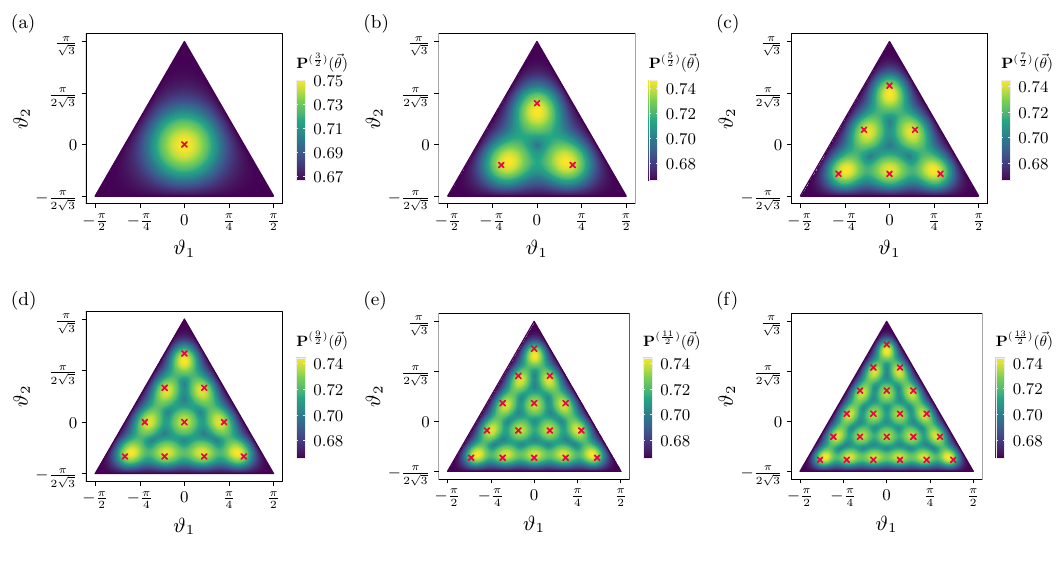}
    \caption{\textbf{Heatmaps of $\mathbf{P}^{(j)}(\vec{\theta})$ for half-integer spins plotted in the $\vec{\vartheta}$ coordinates.} The corresponding spins are (a) $j=3/2$, (b) $j=5/2$, (c) $j=7/2$, (d) $j=9/2$, (e) $j=11/2$, (f) $j=13/2$. Notice that the local peaks follow the triangle numbers $1,3,6,10,\dots$. The resonant angles $\vec{\theta}_{\Delta}^{(j)}(n_1,n_2) = \frac{2\pi}{2j} (n_1,n_2)$ for integers $n_1$ and $n_2$ are also crossed out in red.}
    \label{fig:heatmap-half-spins}
\end{figure*}

\begin{figure*}
    \centering
    \includegraphics[width=\textwidth]{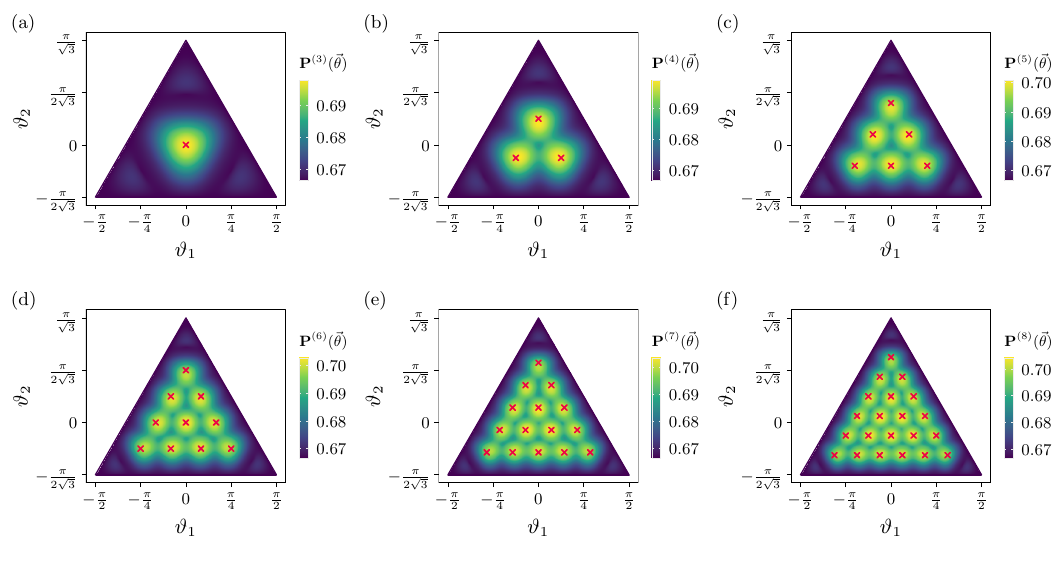}
    \caption{\textbf{Heatmaps of $\mathbf{P}^{(j)}(\vec{\theta})$ for integer spins plotted in the $\vec{\vartheta}$ coordinates.} The corresponding spins are (a) $j=3$, (b) $j=4$, (c) $j=5$, (d) $j=6$, (e) $j=7$, (f) $j=8$. As with Fig.~\ref{fig:heatmap-half-spins}, the local peaks follow the triangle numbers, and the resonant angles $\vec{\theta}_\Delta^{(j)}$ are crossed out in red.}
    \label{fig:heatmap-integer-spins}
\end{figure*}

A pattern becomes apparent by mere visual inspection. For the chosen sequences of spin, we find that there are a triangular number of local maxima, arranged exactly as per its definition as a figurate number. This pattern can be further appreciated by relating the symmetries of each spin particle with the symmetries of the protocol for the choices of angles.

For the equally-spaced protocol, it was previously understood that the large violation of $\mathbf{P}_3^{(3/2)} = 3/4$ by the cat state $\propto\ket{3/2,3/2} - \ket{3/2,-3/2}$ came about because it was ``\emph{in resonance}'' with the three probing times of the original protocol \cite{tsirelson-spin}. That is, under a $J_z$ rotation, states of a spin $j=3/2$ particle can only have discrete symmetries of integer multiples of $2\pi/(2j) = 2\pi/3$ rotations. Since the original protocol respects exactly this symmetry, a large quantum score can be achieved by preparing a state with that symmetry that has a large initial component of $\Theta(J_x)$. This intuition was validated by plotting the spherical Wigner function of the maximally-violating state in Fig.~\ref{fig:Wigner_Resonant}, where it can be clearly seen that the score is augmented by the constructive interference that occurs at angles that are $2\pi/3$ apart.

\begin{figure}
    \centering
    \includegraphics[width=\columnwidth]{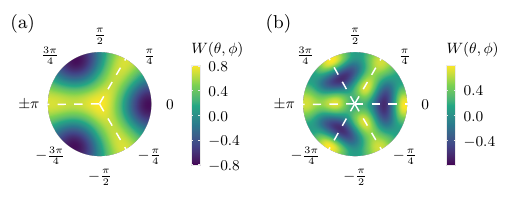}
    \caption{\label{fig:Wigner_Resonant}\textbf{Stereoscopic projection of the $J_z \geq 0$ hemisphere of the spherical Wigner function $W(\theta,\phi)$ \cite{sphericalWigner} of states that achieve the maximum scores.} The plotted states are of (a) a spin-$3/2$ particle with $\mathbf{P}^{(3/2)}_3 = 3/4$ and (b) a spin-$3$ particle with $\mathbf{P}^{(3)}_3 = (8+\sqrt{10})/16 \approx 0.698$. Spherical Wigner functions $W(\theta,\phi)$ are defined on a unit sphere, where $\theta$ and $\phi$ respectively specify the polar and azimuthal angles in spherical coordinates. The constructive interferences that occur at angles that are $2\pi/3$ apart for spin-$3/2$ and $\pi/3$ apart for spin-$3$ augment the scores to enable large violations of the classical bound.}
\end{figure}

Now, if we compare this to the Wigner function of maximally-violating states at the local peaks for other choices of $j$, as with the spin-$7/2$ and spin-$5$ cases plotted in Fig.~\ref{fig:Wigner_Nonresonant}, we observe again that the constructive interference occurs at angles that are about $2\pi/(2j)$ apart. As such, the choices of $\vec{\theta}$ must be such that the probing angles satisfy $\theta_k \approx n_k 2\pi/(2j)$ for integers $n_k$. Together with Eq.~\eqref{eq:condition} and $\theta_1 \leq \theta_2$, this gives the ``\emph{resonant angles}''
\begin{equation}\label{eq:resonance-condition}
    \vec{\theta}_{\Delta}^{(j)}(n_1,n_2) := \pqty{\frac{n_1\pi}{j},\frac{n_2\pi}{j}}
\end{equation}
where $n_1$ and $n_2$ are integers such that $1 \leq n_1 \leq \lfloor j - 1/2 \rfloor$ and $1 + \lfloor j \rfloor \leq n_2 \leq n_1 + \lfloor j - 1/2\rfloor$. For a given $j$, the number of points at which the condition is satisfied at is
\begin{equation}
    \sum_{n_1=1}^{\lfloor j-\frac{1}{2} \rfloor} \sum_{n_2=1 + \lfloor j \rfloor }^{n_1+\lfloor j - \frac{1}{2} \rfloor} = \left\lfloor j - \frac{1}{2} \right\rfloor\pqty{
    \frac{1}{2} + \frac{3}{2}\left\lfloor j - \frac{1}{2} \right\rfloor - \lfloor j \rfloor
    },
\end{equation}
which for the sequences $j = 3/2,5/2,7/2,9/2,\dots$ for the half integer spins and $j=3,4,5,6,\dots$ for the integer spins, give exactly the sequence of triangular numbers $1,3,6,10,\dots$.

Looking back to the protocol scores in Figs.~\ref{fig:heatmap-half-spins}~and~\ref{fig:heatmap-integer-spins}, the local maxima are indeed found in the vicinity of, although not always exactly at, the resonant angles. By maximising over this finite set of local maxima, the global maxima are found close to $\vec{\theta}_{\Delta}^{(j)}(n_G,2n_G)$, and the equivalent choices with respect to the symmetry in Eq.~\eqref{eq:angles-equivalence}, where
\begin{equation}
    n_G = \begin{cases}
        \lfloor (2j+1)/3 \rfloor & \text{for $j$ integer,}\\
        \lfloor j \rfloor & \text{for $j$ half-integer.}
    \end{cases}
\end{equation}

\begin{figure}
    \centering
    \includegraphics[width=\columnwidth]{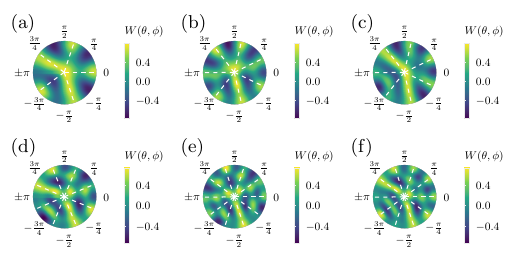}
    \caption{\label{fig:Wigner_Nonresonant}\textbf{Wigner function of the states that achieve the local peaks in Fig.~\ref{fig:heatmap-half-spins}}. These are of the (a) spin-$5/2$, (b,c) spin-$7/2$, (d) spin-$4$, and (e,f) spin-$5$ particles. The constructive interferences occur at angles that are $2\pi/2j$ apart, similarly to the spin-$3/2$ case in Fig.~\ref{fig:Wigner_Resonant}, although the interference patterns are far more complicated.}
\end{figure}

With these observation, we can try to find an even better approximation of the local maxima than the resonant angles. In ``Methods'', we use the latter as initial points in the gradient descent optimisation of $\mathbf{P}^{(j)}(\vec{\theta})$. Let $\vec{\boldsymbol{\theta}}_{\Delta}^{(j)}(n_1,n_2)$ be the local maximum found using gradient descent after starting with the resonant angle $\vec{\theta}_{\Delta}^{(j)}(n_1,n_2)$, where we will leave the arguments $(n_1,n_2)$ implicit when there is no cause for confusion. Then, we heuristically found that $\vec{\boldsymbol{\theta}}_{\Delta}^{(j)} \approx \lambda_j\vec{\theta}_{\Delta}^{(j)} + (1-\lambda_j)\vec{\theta}_3$, where
\begin{equation}
\begin{aligned}
    \frac{1}{\lambda_j} = 1 + \begin{cases}
        \frac{0.533\,051}{j-0.213\,570} & \text{for $j$ integer,}\\
        \frac{0.554\,086}{j-0.197\,425} & \text{for $j$ half-integer.}
    \end{cases}
\end{aligned}
\end{equation}
In other words, the optimal probing angles are approximately the convex combination of the probing angles that reflect the discrete symmetries of the spin system ($\vec{\theta}_{\Delta}^{(j)}$) and those that reflect the discrete symmetries of the protocol ($\vec{\theta}_3$). While our approximate expressions for $\vec{\boldsymbol{\theta}}_{\Delta}^{(j)}$ are heuristic, they may come in handy as approximate values to choose in an experiment, or as starting points for even more precise optimisations.

We are now able to revisit some properties of the obtained scores, referring to Fig.~\ref{fig:originalScores}. In Ref.~\cite{tsirelson-spin}, we had already plotted $\mathbf{P}_3^{(j)}$ against $j$ and observed a damped oscillatory pattern. The damping behaviour was explained by proving analytically that $\lim_{j\to\infty}\mathbf{P}_3^{(j)} = \mathbf{P}_3^{\infty}$, but an explanation for the oscillatory behaviour eluded us. With our new understanding of the location of the local maxima, we see that the oscillatory behaviour came from the fact that the probing angles $\vec{\theta}_3$ of original protocol only corresponds to a resonant angle $\vec{\theta}_{\Delta}^{(j)}$ when $2j$ is a multiple of three. In the other cases, as seen in Figs.~\ref{fig:heatmap-half-spins}~and~\ref{fig:heatmap-integer-spins} where $\vec{\theta}_3$ is the centre point of the triangular points, $\vec{\theta}_3$ is in fact a local minimum.

If we now adapt the angles and plot the local maxima obtained by gradient descent with the initial point at the origin, equivalently given by $\vec{\boldsymbol{\theta}}_{\Delta}^{(j)}(n_0,2n_0)$ with $n_0 = \lfloor (2j+1)/3 \rfloor$, we find that the score as a function of $j$ plots two different smooth sequences, one for half-integer spins and the other for integer spins. Both sequences converge towards $\mathbf{P}_3^{\infty}$. More generally, we observe in Fig.~\ref{fig:largeJparameterSpace} that every point in the interior of $\triangle$ approaches $\mathbf{P}_3^{\infty}$ for large spins, which is due to the convergence $Q^{(j)}(\vec{\theta}) \to Q^{\infty}(\vec{\theta})$ as $j \to \infty$ for any fixed $\vec{\theta}$ \cite{tsirelson-spin}.

\begin{figure*}
    \centering
    \includegraphics{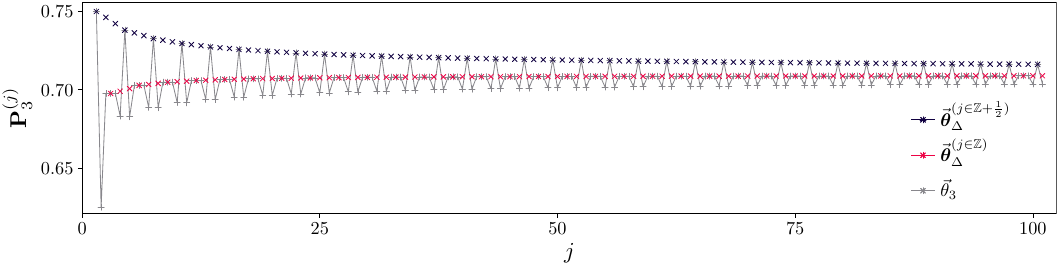}
    \caption{\label{fig:originalScores}\textbf{The scores $\mathbf{P}_3^{(j)}$ of the original equally-spaced precession protocol against $j$.} The scores $\mathbf{P}^{(j)}[\vec{\boldsymbol{\theta}}_\Delta^{(j)}(n_0,2n_0)]$, where $(n_0,2n_0)$ is the local peak closest to $\vec{\theta}_3$, are superimposed. Those corresponding to the half-integer cases $j\in\mathbb{Z+1/2}$ are coloured blue, while those corresponding to the integer cases $j\in\mathbb{Z}$ are coloured red.}
\end{figure*}

\begin{figure}
    \centering
    \includegraphics{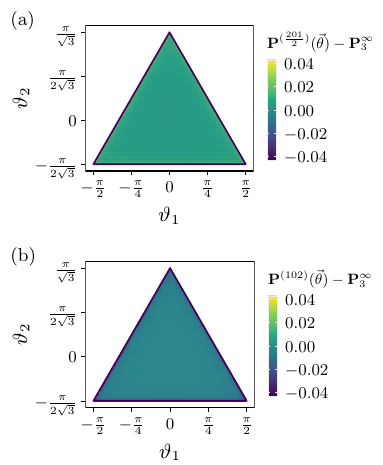}
    \caption{\label{fig:largeJparameterSpace}\textbf{Heatmap of $\mathbf{P}^{(j)}(\vec{\theta}) - 0.709\,364$ for large spins.} They are plotted for (a) half-integer $j=201/2$ and (b) integer $j=102$, where $\mathbf{P}_3^\infty \approx 0.709\,364$ is the approximate value of the maximum score of the quantum harmonic oscillator. The scores $\mathbf{P}^{(j)}(\vec{\theta})$ in the interior of the half-integer spin (integer spin) are all slightly larger (slightly smaller) than but close to $\mathbf{P}_3^\infty$.}
\end{figure}

\subsection{Improved Conjectured Bounds of \texorpdfstring{$\mathbf{P}_3^\infty$}{Maximum Score for the Quantum Harmonic Oscillator}}
Due to the convergence $\lim_{j\to\infty} \mathbf{P}^{(j)}(\vec{\theta}) = \mathbf{P}_3^\infty$, every value $\mathbf{P}^{(j)}(\vec{\theta})$ for sufficiently large $j$ is an approximation for $\mathbf{P}_3^\infty$. Furthermore, we observe that the sequence of $\mathbf{P}^\infty[\vec{\boldsymbol{\theta}}_{\Delta}^{(j)}(n_0,2n_0)]$ for integer $j$ (respectively, half-integer $j$) in Fig.~\ref{fig:originalScores} seem to be monotonically increasing (respectively, monotonically decreasing). We conjecture that this is indeed the case:
\begin{conjecture}
    We conjecture that the integer (half-integer) subsequence of $\mathbf{P}_3^{(3n/2)}$ is monotonically increasing (decreasing) in $n$, that is,
    \begin{equation}
        n \leq n' \implies \begin{cases}
            \mathbf{P}_3^{((2n)3/2)} \leq \mathbf{P}_3^{((2n')3/2)}, \\
            \mathbf{P}_3^{((2n+1)3/2)} \geq \mathbf{P}_3^{((2n'+1)3/2)}.
        \end{cases}
    \end{equation}
\end{conjecture}
Note that we focus on the subsequence where $j$ is a multiple of $3/2$ as $\mathbf{P}^{(3n/2)}[\vec{\boldsymbol{\theta}}^{(3n/2)}_{\Delta}(n_0,2n_0)] = \mathbf{P}_3^{(3n/2)}$ can be more easily computed for large $n$ since $Q_3^{(3n/2)}$ is block-diagonal in that case \cite{tsirelson-spin}.

We checked that this conjecture is true for the first $5500$ terms of both subsequences, as shown in Fig.~\ref{fig:conjecturedBounds}. If this conjecture holds for the for all $j$, combined with the fact that both sequences approach $\mathbf{P}_3^\infty$, this means that $\mathbf{P}_3^{((2n)3/2)} \leq \mathbf{P}_3^\infty \leq \mathbf{P}_3^{((2n'+1)3/2)}$ for any $n,n'\geq0$. This provides us with conjectured bounds $0.709\,364 \leq \mathbf{P}_3^{(225\,000)}  \leq \mathbf{P}_3^\infty \leq \mathbf{P}_3^{(390\,003/2)} \leq 0.709\,511$ for the maximum score of the quantum harmonic oscillator.

Furthermore, we fitted both subsequences with the asymptotic ansatz
\begin{equation}\label{eq:ansatz-spin}
    \mathbf{P}_3^{(j=3n/2)} = \begin{cases}
        \mathbf{P}_3^{\infty} - \sum_{l=1}^2 b_l j^{-l} + \mathcal{O}(j^{-3}) & \text{for $n$ even,} \\
        \mathbf{P}_3^{\infty} + \sum_{l=1}^4 c_l j^{-\frac{l}{2}} + \mathcal{O}(j^{-\frac{5}{2}}) &
        \text{otherwise,}
    \end{cases}
\end{equation}
which is also plotted in Fig.~\ref{fig:conjecturedBounds}. For both numerical fits, the fitted parameters are
\begin{equation}
\begin{aligned}
    \mathbf{P}_3^{\infty} &= 0.709\,364\,176\,0(2)&&\text{for $j=3n$, $n \in 2\mathbb{Z}$} \\
    \mathbf{P}_3^{\infty} &=0.709\,364\,176\,(4)&&\text{for $j=3n + 3/2$, $n \in \mathbb{Z}$}
\end{aligned}
\end{equation}
up to the standard error of the fit reported in the parentheses. Both agree with the numerical fit
\begin{equation}
    \mathbf{P}_3^{\infty} =
    0.709\,364\,176\,01(5)
\end{equation}
of the rigorous lower bound in Fig.~\ref{fig:lowerBound} to nine decimal places. All in all, the numerical evidence strongly suggests that the true value of the maximum score of the quantum harmonic oscillator is indeed $\mathbf{P}_3^{\infty} \approx 0.709\,364$.

\begin{figure}
    \centering
    \includegraphics[width=\columnwidth]{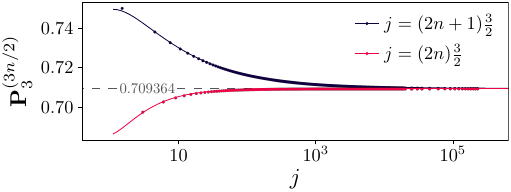}
    \caption{\label{fig:conjecturedBounds}\textbf{The score $\mathbf{P}_3^{(j)}$ when $2j$ is a multiple of $3$.} We conjecture that $\mathbf{P}_3^{(3n/2)}$ is monotonically increasing (decreasing) for $n$ even ($n$ odd), so their convergence to $\mathbf{P}_3^\infty$ would imply that their values for $n$ even ($n$ odd) form a sequence of lower bounds (upper bounds) of $\mathbf{P}_3^\infty$. We also perform a numerical fit to the ansatz given in Eq.~\eqref{eq:ansatz-spin}, which implies that both sequences converge to $\lim_{n\to\infty}\mathbf{P}_3^{(3n/2)} = \mathbf{P}_3^\infty = 0.709\,364$.}
\end{figure}

\subsection{\label{secentg}Implications on Detecting Entanglement}

The original precession protocol has been shown to be useful for detecting entanglement when performed on a collective coordinate of the system. In this section, we show how these previous results carry over to generalisations of the precession protocol.

Reference~\cite{tsirelson-harmonic-entanglement} considered two coupled harmonic oscillators, with local coordinates $(X_n,P_n)$ for the $n$th oscillator, of the form
\begin{equation}
\begin{aligned}
    H &= \sum_{n=1}^2 \frac{\omega_n}{2} \pqty{X_n^2 + P_n^2} - \frac{g}{2} X_1X_2 \\
    &= \sum_{\sigma \in \{+,-\}} \frac{\omega_\sigma}{2} \pqty{X_{\sigma\varphi}^2 + P_{\sigma\varphi}^2},
\end{aligned}
\end{equation}
with collective coordinates $X_{+\varphi} = X_1 \cos\varphi + X_2\sin\varphi$ and $X_{-\varphi} = X_2\cos\varphi - X_1\sin\varphi$, similarly defined for $P_{\pm\varphi}$, where $\varphi = \atantwo(g,\omega_1^2-\omega_2^2)$. By performing the precession protocol on the collective coordinate $X_{\sigma\varphi}(t)$ for $\sigma \in \{+,-\}$, entanglement between the local coordinates can be detected when the obtained score $P_{3,\sigma}^\infty(\varphi) := {\tr}[\rho Q_{3,\sigma}^\infty(\varphi)]$, where $3Q_{3,\sigma}^\infty(\varphi) := \sum_{k=0}^2 \Theta[X_{\sigma\varphi}({2\pi k}/{3\omega_\sigma})]$, satisfies
\begin{equation}
P_{3,\sigma}^\infty(\varphi) > \mathbf{P}^{\infty\text{-sep}}_3(\varphi) := \max_{\varrho \in\operatorname{PPT}}\tr[ \varrho Q_{3,\sigma}^\infty(\varphi)].
\end{equation}
Here, $\operatorname{PPT} := \{\varrho : \varrho^{T_2}\succeq 0\}$ is the set of positive partial transpose states, where $\rho^{T_2}$ is the partial transpose over the $(X_2,P_2)$ mode. The special case $\mathbf{P}^{\infty\text{-sep}}_3(\pi/4) = \mathbf{P}_3^c$ is known analytically due to the relationship between Wigner negativity and entanglement \cite{tsirelson-harmonic-entanglement,CD-entanglement,phase-space-entanglement} (see in particular Theorem~2 of Ref.~\cite{CD-entanglement}), while the general case can be calculated using semidefinite programming by truncating the energy levels \cite{tsirelson-harmonic-entanglement}.

The same idea can be applied to the family of precession protocols with three angles. For a choice of probing time $\vec{\theta}$, we can similarly perform the protocol on the collective coordinate $X_{\sigma\varphi}(t)$ to obtain the score $P_\sigma^\infty(\varphi,\vec{\theta}) := {\tr}[\rho Q_{\sigma}^\infty(\varphi,\vec{\theta})]$, where $3 Q_{\sigma}^\infty(\varphi,\vec{\theta}) := \sum_{k=0}^2 \Theta[X_{\sigma\varphi}({\theta_k}/{\omega_\sigma})]$. Then, entanglement is detected when $P_\sigma^\infty(\varphi,\vec{\theta}) > \mathbf{P}^{\infty\text{-sep}}(\varphi,\vec{\theta})$, where the separable bound is similarly defined as
\begin{equation}
    \mathbf{P}^{\infty\text{-sep}}(\varphi,\vec{\theta}) := \max_{\varrho\in\operatorname{PPT}}\tr[ \varrho Q_{\sigma}^\infty(\varphi,\vec{\theta})].
\end{equation}
In the Supplementary Information, we show that for all points in the interior region $\triangle$ such that $\mathbf{P}^c(\vec{\theta}) = 2/3$, we have $\mathbf{P}^{\infty\text{-sep}}(\varphi,\vec{\theta}) = \mathbf{P}^{\infty\text{-sep}}_3(\varphi)$. This means that the separable bounds previously found in Ref.~\cite{tsirelson-harmonic-entanglement} also hold for all three angle precession protocols performed on a collective mode. In particular, $\mathbf{P}^{\infty\text{-sep}}(\pi/4,\vec{\theta}) = \mathbf{P}_3^c$ for all $\vec{\theta}$ in the interior.

These results allow us to witness the non-Gaussian entanglement of states that could not be detected with the original protocol. For example, consider the state
\begin{equation}
\begin{aligned}
    \ket{\chi_4} = \frac{1}{\sqrt{10}}\Big[
    &2 \cos({10}/{17})\ket{0}_+ - \sqrt{5}\sin({2}/{3})\ket{1}_+ \\
    &{}+{} 2 \sin({10}/{17}) \ket{2}_+ - \sqrt{5} \cos({2}/{3}) \ket{3}_+ \\
    &{}+{} \ket{4}_+
    \Big] \otimes \ket{0}_-,
\end{aligned}
\end{equation}
where $\ket{n}_\pm$ are the number states in the collective mode $(X_{\pm\pi/4},P_{\pm\pi/4})$. $\ket{\chi_4}$ does not violate the original protocol as there are no states truncated to the first five energy levels for which $P_3^\infty > \mathbf{P}_3^c$ \cite{tsirelson-anharmonic}. However, $P^\infty_\sigma[\pi/4,(\pi^2/4,\pi^2/2)] = 0.669 > \mathbf{P}_3^c$ when performed on the collective mode $X_{+\pi/4}$, so we can detect the entanglement of $\ket{\chi_4}$ with a different set of probing angles.

Another important application is the detection of squeezed versions of states. Recently, advances in bosonic error correction has led to the development of squeezed cat codes of the form $\propto \sum_{k=0}^{K-1} S_\lambda \lvert\alpha e^{i\frac{2\pi k}{K}}\rangle$ where $S_\lambda := \exp[-i\lambda(XP+PX)/2]$ is the squeeze operator, which is robust against a variety of error sources \cite{squeezedCat1,squeezedCat2}. Previously, it was found that the entanglement of the entangled tricat state $\ket{\mathrm{cat}_3(\alpha)}\propto\sum_{k=0}^2 \lvert\alpha e^{i\frac{2\pi k}{3}}\rangle\otimes\lvert\alpha e^{i\frac{2\pi k}{3}}\rangle$ could be detected with the original precession protocol for certain values of $\alpha$ for which $P^\infty_{3,+}(\pi/4) > \mathbf{P}_3^c$ \cite{tsirelson-harmonic-entanglement}. From the relation $S_\lambda^{\otimes 2} Q_{3,+}^\infty(\pi/4) S_\lambda^{\dag \otimes 2} = Q_{+}^\infty[\pi/4,(\tilde{\theta}_-,\tilde{\theta}_+)]$ that comes from the constructive proof of Supplementary Theorem~3, where $\tilde{\theta}_\mp = \pi\mp\atan(\sqrt{3}e^{2\lambda})$, the squeezed version of the cat state $S_\lambda^{\otimes 2}\lvert{\mathrm{cat}_3(\alpha)}\rangle$ can be detected by performing the precession protocol on a collective mode with probing angles $\vec{\theta} = (\tilde{\theta}_-,\tilde{\theta}_+)$, since
\begin{equation}
\begin{aligned}
    &P^\infty_{+}[\pi/4,(\tilde{\theta}_-,\tilde{\theta}_+)] \\
    &\quad{}={}  \langle{\mathrm{cat}_3(\alpha)}\rvert  S_\lambda^{\dag \otimes 2} Q_{+}^\infty[\pi/4,(\tilde{\theta}_-,\tilde{\theta}_+)] S_\lambda^{\otimes 2} \lvert{\mathrm{cat}_3(\alpha)}\rangle \\
    &\quad{}={} \langle{\mathrm{cat}_3(\alpha)}\rvert Q_{3,+}^\infty(\pi/4)  \lvert{\mathrm{cat}_3(\alpha)}\rangle\\
    &\quad{}={} P^\infty_{3,+}(\pi/4) > \mathbf{P}_3^c.
\end{aligned}
\end{equation}

Meanwhile, Ref.~\cite{tsirelson-spin-entanglement} considered ensembles of $N$ particles with fixed spins $\{j_n\}_{n=1}^N$, where $j_n$ is the spin and $\vec{J}^{(j_n)} = (J_x^{(j_n)},J_y^{(j_n)},J_z^{(j_n)})$ the angular momentum of the $n$th particle. There, it was shown that for any observable defined by a function $f(\vec{J})$ of the total angular momentum $\vec{J} = (J_x,J_y,J_z) = \sum_{n=1}^N \vec{J}^{(j_n)}$ of the system
\begin{equation}
\begin{aligned}
\tr[\rho f(\vec{J})] &> \max_{j+j' \leq \sum_{n=1}^N j_n}
    \max_{\varrho \in \operatorname{PPT}_{j,j'}}\tr[\varrho f\pqty{\vec{J}^{(j)} + \vec{J}^{(j')} }] \\
&\implies \rho\text{ is genuinely multipartite entangled,}
\end{aligned}
\end{equation}
where $\operatorname{PPT}_{j,j'}$ is the set of positive partial transpose states over the tensor product of a spin-$j$ and spin-$j'$ system. Here, $\rho$ being genuinely multipartite entangled (GME) means that $\rho$ is not a probabilistic mixture of states separable over any bipartition of the $N$ spins.

One can then perform the precession protocol on the total angular momentum of the spin ensemble with any choice of probing angles, for which the score will be given by $P(\vec{\theta}) := {\tr}[\rho Q^{(\{j_n\}_{n})}(\vec{\theta})]$ where $3Q^{(\{j_n\}_{n})}(\vec{\theta}) := \sum_{k=0}^n \Theta[J_x\cos\theta_k + J_y\sin\theta_k]$. Since $Q^{(\{j_n\}_{n})}(\vec{\theta})$ is a function of the total angular momentum, we can define the separable bound as
\begin{equation}
\begin{aligned}
    \mathbf{P}^{\operatorname{sep}}(\{j_n\}_{n=1}^N,\vec{\theta}) &:= \hspace{-1em}\max_{j,j' \leq \sum_{n}j_n}  \underbrace{\max_{\varrho \in \operatorname{PPT}_{j,j'}}\tr[\varrho Q^{(\{j,j'\})}(\vec{\theta}) ]}_{\mathbf{P}^{\operatorname{sep}}(\{j,j'\},\vec{\theta})},
\end{aligned}
\end{equation}
which can be written as a maximisation over biseparable bounds $\mathbf{P}^{\operatorname{sep}}(\{j,j'\},\vec{\theta})$. This can in turn be calculated using semidefinite programming, whose results are shown in Fig.~\ref{fig:heatmap_SDP}. We found that $\mathbf{P}^{\operatorname{sep}}(\{j_n\}_{n=1}^N,\vec{\theta}) \leq \mathbf{P}_3^c$ when $\sum_n j_n \leq 2$. Hence, whenever a spin ensemble with total spin $\leq 2$ violates the classical bound, its constituents must be GME.

\begin{figure}
    \includegraphics[width=\columnwidth]{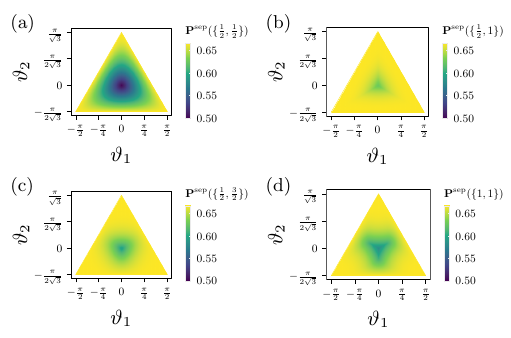}
    \caption{\label{fig:heatmap_SDP}\textbf{Biseparable bounds $\mathbf{P}^{\text{sep}}(\{j,j'\},\vec{\theta})$.} They are plotted for (a) $\{j,j'\}=\{1/2,1/2\}$, (b) $\{j,j'\}=\{1/2,1\}$, (c) $\{j,j'\}=\{1/2,3/2\}$, and (d) $\{j,j'\}=\{1,1\}$. For $j+j' \leq 2$, we find that the biseparable bounds satisfy $\mathbf{P}^{\text{sep}}(\{j,j'\},\vec{\theta})\leq \mathbf{P}_3^c$. As such, $\mathbf{P}^{\text{sep}}(\{j_n\}_n,\vec{\theta}) \leq \mathbf{P}_3^c$ for $\sum_{n}j_n \leq 2$, so GME is implied when the classical bound is violated upon performing any of the three-angle generalised precession protocols on the total angular momentum of a system with total spin at most two.}
\end{figure}

We are also able to detect more states with this larger family of GME witnesses. For example, consider the state $e^{-i(49\pi/60)J_z}\ket{\psi_4}$ of a four-particle ensemble of spin-$1/2$ particles, where
\begin{equation}
\begin{aligned}
    \ket{\psi_4} &= \frac{\sqrt{6}+1}{6}\ket{\Phi_+}^{\otimes2} + \frac{\sqrt{6}-1}{6}\ket{\Phi_-}^{\otimes2} + \frac{1}{3}\ket{\Psi_+}^{\otimes2} \\
    &\qquad{}+{} \frac{1}{2}\pqty\big{\ket{\Phi_+}\otimes\ket{\Psi_+} + \ket{\Psi_+}\otimes\ket{\Phi_+}},
\end{aligned}
\end{equation}
such that $\ket{\Psi_\pm} \propto \ket{\uparrow\downarrow} \pm \ket{\downarrow\uparrow}$ and $\ket{\Phi_{\pm}} \propto \ket{\uparrow\uparrow} \pm \ket{\downarrow\downarrow}$ are the Bell states. This state achieves the score $P_3 = 0.6 < \mathbf{P}_3^c$ for the original protocol performed on the total angular momentum, and the score $P((49\pi/60,49\pi/30)) = 0.67 > \mathbf{P}_3^c$ for the modified protocol. Therefore, the GME of $e^{-i(49\pi/60)J_z}\ket{\psi_4}$ can only be detected by the modified protocol and not the original one.

\subsection{\label{sec:finitestats}Impact of Limited Measurement Precision}
As this work is focused on the theoretical study of the precession protocol, we have thus far assumed that the measured values of $A_x(\theta)$ are precise enough for the unambiguious assignment of $\Theta[A_x(\theta)] \in \{0,1/2,1\}$ when calculating the score $P(\vec{\theta})$. This certainly holds for discrete variable systems, like the angular momentum of particles with finite spin, where the outcomes of $A_x$ are discrete and sufficiently separated from one another. In such systems, $\Theta[A_x(\theta)]$ can be determined exactly as long as the finite gap between $0$ and its neighbouring values can be resolved by the measurement apparatus.

However, in actual experimental implementations of continuous variable systems, the measurement device used to measure $A_x(\theta)$ might have a limited resolution. This would mean that the measured values of $A_x(\theta)$ in the vicinity of $A_x(\theta) \approx 0$ might not be precise enough for $\Theta[A_x(\theta)]$ to be determined unambiguiously.

More generally, let us consider some coarse graining of continuous variables into discrete bins specified by their end points $\{a_{n}\}_{n \in \mathbb{Z}}$, where the $n$th bin contains the measurement outcomes $a_{n} \leq A_x(\theta) < a_{n+1}$, with bin width $a_{n+1}-a_n$. Each bin might correspond to a particular division of a physical rule or a particular number with a finite number of decimal places, which are both examples of the limited resolution of the measurement apparatus. The experimentalist would therefore only know which bin the outcome was in, and not the exact value of $A_x(\theta)$ to infinite precision.

Label the $0$th bin to be the bin that contains $0$ such that $a_{0} \leq 0 < a_{1}$, and the $\mathbf{n}$th bin to be any bin chosen such that the total width of the $1$st to $\mathbf{n}$th bin $\sum_{n=1}^{\mathbf{n}} (a_{n + 1} - a_{n}) = a_{\mathbf{n}+1} - a_1 > a_{1}-a_0$ is larger than the width of the $0$th bin. If all bins have the same width, it suffices to choose $\mathbf{n} = 1$.

Then, we show in the Supplementary Information that given the modified score assignment $\Theta[A_x(\theta)] \to \widetilde{\Theta}[A_x(\theta)]$, where
\begin{equation}
    \widetilde{\Theta}[A_x(\theta)] := \begin{cases}
        1 & \text{if $A_x(\theta) \in n$th bin with $n > \mathbf{n}$,} \\
        1/2 & \text{if $A_x(\theta) \in n$th bin with $0 \leq n \leq \mathbf{n}$,} \\
        0 & \text{otherwise,}
    \end{cases}
\end{equation}
the classical inequality $P^c(\vec{\theta}) \leq \mathbf{P}_3^c$ still holds for $\vec{\theta} \in \triangle$. Therefore, when the measurement apparatus has limited precision, the score assignment of the coarse-grained data as given above ensures that no false positives occur, such that a violation of the same classical bound $\mathbf{P}_3^c$ still certifies nonclassicality.

\subsection{\label{secmore}Implications and Outlook on Protocols with More Angles}
In Ref.~\cite{tsirelson-spin}, the precession protocol was generalised to $K$ equally-spaced angles for odd $K$---note that there is also a related protocol for even $K$ \cite{tsirelson-even}, but it is rather different and not discussed in this paper. Just like the preceding sections, we consider now the generalisation to $K$ arbitrarily-spaced angles. The $K$ measured angles can be labelled such that $0=\theta_0 \leq \theta_1 \leq \dots \leq \theta_{K-1} \leq 2\pi$, again up to an offset, so a particular choice of probing angles is specified by the vector $\vec{\theta} := (\theta_k)_{k=1}^{K-1}$.

In the Supplementary Information, we show that
$\mathbf{P}^c(\vec{\theta}) = 1-\delta/K$ for $\delta \in \{0,1,\dots,(K-1)/2\}$ if and only if
\begin{equation}
\forall k: \theta_{(k\,\oplus_{_{\!K}}\delta)} \ominus_{_{\!2\pi}} \theta_k \leq \pi,
\end{equation}
with a reminder that $x\,\ooalign{\hidewidth$\pm$\hidewidth\crcr$\bigcirc$}_{\!m}\, y = (x\pm y)\bmod m$ are addition and subtraction modulo $m$, respectively. The parameter space $\vec{\theta}$ is therefore split into $(K+1)/2$ regions where $\mathbf{P}^c \in \{1-\delta/K\}_{\delta=0}^{(K-1)/2}$. For $K=3$, these are precisely triangular regions with $\mathbf{P}^c = 2/3$ and $\mathbf{P}^c = 1$, respectively, but the regions become more complex to characterise for $K>3$.

For the $K=3$ case for the quantum harmonic oscillator, the scores in the interior region $\triangle$ are all $\mathbf{P}^\infty(\vec{\theta}) = \mathbf{P}_3^\infty$ because every observable $Q^\infty(\vec{\theta})$ within the same region can be symplectically transformed to each other. However, because there are only three degrees of freedom---change of phase, magnitude of squeezing, axis of squeezing---once $3$ out of the $K$ angles are transformed, the other $K-3$ angles are fixed. Therefore, for $K>3$ the different observables $Q^\infty(\vec{\theta})$ within the same region are no longer symplectically related, and so the maximum quantum score $\mathbf{P}^\infty(\vec{\theta})$ is no longer the same within the region with the same classical score.

Meanwhile, the upper bound $\mathbf{P}^{\geq}(\vec{\theta}) \geq \mathbf{P}^\infty(\vec{\theta})$ within the region where $\mathbf{P}^{c}(\vec{\theta}) = (1+1/K)/2$ can be found by extending the derivation of the upper bound $\mathbf{P}^{\geq}_3$. In principle, closed-form expressions of $\mathbf{P}^{\geq}(\vec{\theta})$ can be obtained, and the steps required to do so are laid out in the Supplementary Information. However, the obtained expressions are extremely cumbersome for $K>3$, so we have instead evaluated $\mathbf{P}^{\geq}(\vec{\theta})$ numerically. This only involved integrals over piecewise smooth functions, and thus the calculated bounds, although obtained numerically, are reliable.

In Fig.~\ref{fig:generalBounds}, we plot the upper bound $\mathbf{P}_K^{\geq}$ of the maximum quantum score $\mathbf{P}_K^\infty$ for the equally-spaced protocol $\vec{\theta}_K := (2\pi k/K)_{k=1}^{K-1}$. The best upper bound previously known was $\mathbf{P}_K^\infty \leq \mathbf{P}^c(\vec{\theta}_K) + 0.155\,940$, where $\mathbf{P}^c(\vec{\theta}_K) = (1+1/K)/2$, which again came from bounds of the single-wedge integral of Wigner functions \cite{tsirelson-spin}. However, the old upper bound saturates at $1/2 + 0.155\,940$ as $K\to \infty$, which is the incorrect asymptotic behaviour, as it is known that $\lim_{K\to\infty}\mathbf{P}_K^\infty = 1/2$ \cite{tsirelson-spin}. Rather, the new upper bound $\mathbf{P}_K^{\geq}$ not only exhibits the correct asymptotic behaviour, but is also close to the previously-known lower bound $\mathbf{P}_K^{\leq}$ \cite{tsirelson-spin}. Furthermore, the new upper bound in turn provides an improved upper bound of
\begin{equation}
    \sup_\rho \abs{\frac{1}{2} - \iint_{\Omega_K}\dd{x}\dd{p}W_\rho(x,p)} \leq  K\pqty{\mathbf{P}_K^{\geq}-\frac{1}{2}},
\end{equation}
where $\Omega_K = \{(x,p): \atantwo(p,x) \bmod (4\pi/K) \leq 2\pi/K \}$ is the equally-spaced $K$-wedge region.

\begin{figure}
    \centering
    \includegraphics[width=\columnwidth]{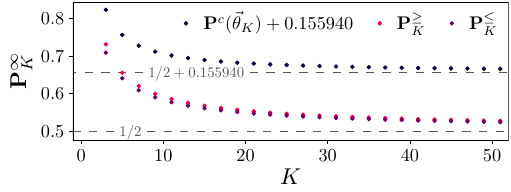}
    \caption{\label{fig:generalBounds} \textbf{Plot of the bounds $\mathbf{P}_K^{\leq} \leq \mathbf{P}_K^{\infty} \leq \mathbf{P}_K^{\geq}$ for the quantum harmonic oscillator, and for the equally-spaced protocol.} A comparison to the previously known upper bound $\mathbf{P}^c(\vec{\theta}_K) + 0.155\,940$ \cite{tsirelson-spin}  is also included. The improved upper bound exhibits the correct asymptotic behaviour for $\lim_{K\to\infty}\mathbf{P}_K^{\infty} = 1/2$, and is also very close to $\mathbf{P}_K^{\leq}$.}
\end{figure}

Lastly, we consider the precession protocol with more probing angles applied to spin angular momentum. While we can no longer plot out $\mathbf{P}^{(j)}(\vec{\theta})$ for the full parameter space like we did for $K=3$, we can try to observe some properties of the local peaks by starting with random initial parameters $\vec{\theta}$ and performing gradient descent. Some qualitative behaviour carries over from the $K=3$ case. For example, for $K=5$, we notice that many local peaks occur for $\vec{\boldsymbol{\theta}}^{(j)}_\Delta \approx \vec{\theta}_\lambda^{(j)} := (1-\lambda)\vec{\theta}^{(j)}_{\Delta} + \lambda \vec{\theta}_K$ region, where $\vec{\theta}^{(j)}_{\Delta} = ( \pi n_k/j )_{k=1}^{K-1}$, and $\{n_k\}_{k=1}^{K-1}$ are integers such that
\begin{equation}
0 < n_1 < \dots < n_{\frac{K-1}{2}} \leq \lfloor j \rfloor < n_{\frac{K+1}{2}} < \dots \leq K-1.
\end{equation}
This is shown in the Supplementary Information for $j \in \{7/2,11/2\}$. Although we also find local maxima not of the form $\vec{\theta}_\lambda^{(j)}$, we nonetheless observe that choosing initial points of the form $(\vec{\theta}^{(j)}_{\Delta} + \vec{\theta}_K)/2$ for the gradient descent provides faster convergence than starting with random initial points, and that the chosen initial points already achieve large violations of the classical bound. This aids in finding states and parameters $\vec{\theta}$ that obtain large scores $\mathbf{P}^{(j)}(\vec{\theta})$, similar to those prepared in recent experimental implementations of the precession protocol \cite{tsirelson-experiment}.

\section{\label{secconcl}Discussions}
In this work, we have characterised the family of Tsirelson's precession protocol with three angles, which include as special cases some previously-studied generalisations like the Type I and II Tsirelson inequalities. We answer an open question about the maximum violation possible for these inequalities by showing that the maximum score for the harmonic oscillator is the same for all members of the characterised family. Furthermore, we provide new rigorous and conjectured bounds for the maximum quantum score $\mathbf{P}_3^\infty$ of the original protocol that improve upon the best bounds that were previously known, which also contribute improved bounds of integrals of Wigner functions over certain phase space regions. Finally, by extrapolating the rigorous and conjectured bounds, we estimate that the true value of the maximum quantum score is $\mathbf{P}_3^\infty = 0.709\,364$.

We also studied the family of protocols when applied to spin angular momentum. There, we found that the quantum violation can be increased for a given spin system by adjusting the probing angles, and observed that the location of the local maxima in parameter space followed the pattern of triangle numbers. The latter observation came from the fact that the optimal probing angles are close to a mixture of probing angles that reflect the symmetry of the protocol and that of the given spin, which also gives us an approximation of the optimal probing angles. This will be useful for choosing the probing angles in experimental implementations of the precession protocol, as recently demonstrated \cite{tsirelson-experiment}.

Afterwards, we related our findings to some previous results on the connection between the precession protocol and witnesses of entanglement. We show that similar to the original protocol, every member of the generalised family is an entanglement witness when applied to the collective coordinate of coupled harmonic oscillators or spin ensembles. For the coupled harmonic oscillator, we further show that the separable bounds for every generalised protocol is the same as that of the original protocol. For both systems, we show by explicit examples that there are entangled states that can only be detected with the generalised protocol and not the original one. Therefore, this work also introduces new non-Gaussian and genuine multipartite entanglement witnesses.

Finally, we briefly touched upon precession protocols with more than three angles. We show that some results---like the classical bound and the upper bound of the maximum quantum score---can be generalised to protocols with more angles, but also comment upon some results---like the invariance of the maximum quantum score for the harmonic oscillator---that cannot. We further demonstrate how observations from the three angle case can inform strategies for finding larger violations for the precession protocol with more angles.

There are some evident future directions: the first is to extend our work to larger $K$. The complexity grows very quickly, as the parameter space of $K=5$ is already four dimensional, with the different regions taking more complicated forms. The second is to extend this work to general theories, as has been done for the original protocol \cite{tsirelson-general}. We note that our proof that the maximum harmonic oscillator score is the same for all three-angle generalisations depend only on symplectic transformations on the measured observables. Thus, our result also holds for the recently-introduced general theories based on continuous variable quasiprobability distributions, since these distributions also permit symplectic transformations \cite{GPT-CV-1,GPT-CV-2}.

\section{Methods}
\subsection{\label{apd:score-gradient}Heuristic Optimisation of Local Maxima}
We observed in ``Spin Angular Momentum'' that the local maxima are close to, but not exactly at, the resonant angles $\vec{\theta}_{\Delta}^{(j)}(n_1,n_2) = (n_1\pi/j,n_2\pi/j)$. We numerically found the local maxima using $\vec{\theta}_{\Delta}^{(j)}$ as the initial point for the gradient descent optimisation of $\mathbf{P}^{(j)}(\vec{\theta})$. This requires the gradient of $\mathbf{P}^{(j)}(\vec{\theta})$, which is given by
\begin{widetext}
\begin{equation}\label{eq:gradient}
\begin{aligned}
    \frac{\partial\mathbf{P}^{(j)}(\vec{\theta})}{\partial\theta_k}
    &= i\frac{1}{3} \langle{\mathbf{P}^{(j)}(\vec{\theta})}\rvert e^{-i\theta_k J_z} \comm{\Theta(J_x)}{J_z} e^{i\theta_k J_z} \lvert{\mathbf{P}^{(j)}(\vec{\theta})}\rangle \\
    &=\begin{cases}
    \begin{array}{l}
     i\frac{\sqrt{j(j+1)}}{12}
    \langle{\mathbf{P}^{(j)}(\vec{\theta})}\rvert e^{-i\theta_k J_z} e^{-i\frac{\pi}{2} J_y} \big[
    \ketbra{j,-1}{j,0} + \ketbra{j,0}{j,1} \\
    \hspace{15em} {}-{}
        \ketbra{j,1}{j,0} - \ketbra{j,0}{j,-1}
    \big] e^{i\frac{\pi}{2} J_y} e^{i\theta_k J_z} \lvert{\mathbf{P}^{(j)}(\vec{\theta})}\rangle
    \end{array}
    & \text{for $j$ integer,}\\[1.5ex]
    i\frac{j+1/2}{6}
        \langle{\mathbf{P}^{(j)}(\vec{\theta})}\rvert e^{-i\theta_k J_z} e^{-i\frac{\pi}{2} J_y} \pqty{
        \ketbra{j,-\tfrac{1}{2}}{j,\tfrac{1}{2}} -
        \ketbra{j,\tfrac{1}{2}}{j,-\tfrac{1}{2}}
    } e^{i\frac{\pi}{2} J_y} e^{i\theta_k J_z} \lvert{\mathbf{P}^{(j)}(\vec{\theta})}\rangle & \text{otherwise,}
    \end{cases}
\end{aligned}
\end{equation}
where $\lvert{\mathbf{P}^{(j)}(\vec{\theta})}\rangle$ is the maximal eigenstate of $Q^{(j)}(\vec{\theta})$. Here, the analytical expression for $\comm{\Theta(J_x)}{J_z}$ was found in a previous work \cite{tsirelson-spin}.
\end{widetext}

Equation~\eqref{eq:gradient} is passed into standard gradient descent libraries to maximise the objective $\mathbf{P}^{(j)}(\vec{\theta})$. The local peaks $\vec{\boldsymbol{\vartheta}}_{\Delta}^{(j)}$ as found are plotted alongside $\vec{\vartheta}_{\Delta}^{(j)}$ in Fig.~\ref{fig:angles}, where they appear to be scalar multiples of each other. To investigate this further, we plotted the angle
\begin{equation}
\angle(\vec{\boldsymbol{\vartheta}}_{\Delta}^{(j)},\vec{\vartheta}_{\Delta}^{(j)}) = \acos(
    \frac{
        \vec{\boldsymbol{\vartheta}}_{\Delta }^{(j)}\cdot\vec{\vartheta}_{ \Delta }^{(j)}
    }{
        |\vec{\boldsymbol{\vartheta}}_{\Delta }^{(j)}||\vec{\vartheta}_{ \Delta }^{(j)}|
    }
)
\end{equation}
for a large range of $j$ in Fig.~\ref{fig:angles}. $\angle(\vec{\boldsymbol{\vartheta}}_{\Delta }^{(j)},\vec{\vartheta}_{ \Delta }^{(j)})\approx 0$ for all values that we calculated, which verifies that they are indeed roughly scalar multiples of each other.

\begin{figure}
    \centering
    \includegraphics[width=\columnwidth]{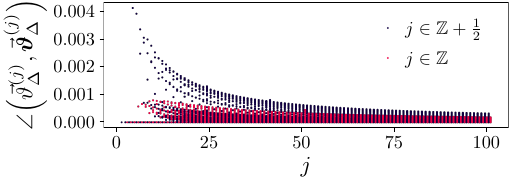}
    \caption{\label{fig:angles}\textbf{The angle $\angle(\vec{\boldsymbol{\vartheta}}_{\Delta}^{(j)},\vec{\vartheta}_\Delta^{(j)})$ between the local peaks $\vec{\boldsymbol{\vartheta}}^{(j)}_{\Delta}$ and the resonant probing angles $\vec{\vartheta}_{\Delta}^{(j)}$.} We find that $\angle(\vec{\boldsymbol{\vartheta}}_{\Delta}^{(j)},\vec{\vartheta}_\Delta^{(j)}) \approx 0$, which shows that $\vec{\boldsymbol{\vartheta}}_{\Delta}^{(j)}$ and $\vec{\vartheta}_{\Delta}$ are approximately scalar multiples of each other.}
\end{figure}

With a phenomenological fit of the scaling factor, we also found $\vec{\boldsymbol{\vartheta}}_{ \Delta }^{(j)} \approx \lambda_j\vec{\vartheta}_{ \Delta }^{(j)}$ to be an excellent approximation of the local maximum, where
\begin{equation}
\begin{aligned}
    \frac{1}{\lambda_j} = 1 + \begin{cases}
        \frac{0.533\,051}{j-0.213\,570} & \text{for $j$ integer,}\\
        \frac{0.554\,086}{j-0.197\,425} & \text{otherwise.}
    \end{cases}
\end{aligned}
\end{equation}
The discrepancy is at most $|\vec{\boldsymbol{\vartheta}}^{(j)}_{\Delta} - \lambda_j\vec{\vartheta}^{(j)}_{\Delta}| \leq 0.004$ for $j \leq 100$, which decreases monotonically with $j$, as shown in Fig.~\ref{fig:errors}.

\begin{figure}
    \centering
    \includegraphics[width=\columnwidth]{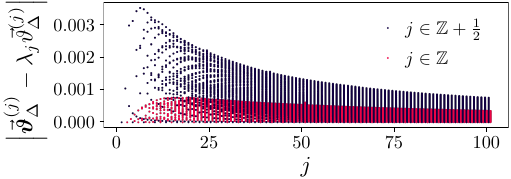}
    \caption{\label{fig:errors}\textbf{The discrepancy $|\vec{\boldsymbol{\vartheta}}_{\Delta} - \lambda_j\vec{\vartheta}_{\Delta}|$ between the local peaks $\vec{\boldsymbol{\vartheta}}_{\Delta}$ and the rescaled resonant probing angles $\lambda_j\vec{\vartheta}_{\Delta}$.} We find $\lambda_j\vec{\vartheta}_{\Delta}$ to be an excellent approximation for $\vec{\boldsymbol{\vartheta}}_{\Delta}$.}
\end{figure}

\section*{Data Availability}
The data generated in this study utilise only standard numerical techniques (i.e., eigenvalue solvers, semidefinite programming), which can be replicated from the provided equations. The raw data is also available from the corresponding author upon reasonable request.

\section*{Code Availability}
The code used is available from the corresponding author upon reasonable request.

\section*{Author Contributions}
L.H.Z. prepared all figures and wrote the first draft, V.S. supervised and edited the final version. All authors reviewed the manuscript.

\section*{Competing Interests}
The authors declare no competing interests.

\begin{acknowledgments}
This work is supported by the National Research Foundation, Singapore, and A*STAR under its CQT Bridging Grant; and by the same National Research Foundation, Singapore through the National Quantum Office, hosted in A*STAR, under its Centre for Quantum Technologies Funding Initiative (S24Q2d0009). The computation involved in this work is supported by NUS IT's Research Computing group under the grant NUSREC-HPC-00001. L.H.Z. thanks Harkan Kirk-Karakaya for pointing out that the explicit proof of the connection to the backflow constant given in Eq.~\eqref{eq:backflow} was missing in an earlier version of the manuscript.
\end{acknowledgments}

\bibliography{refs}

\clearpage

\phantomsection\addcontentsline{toc}{part}{Supplementary Information for: All three-angle variants of Tsirelson's precession protocol, and improved bounds for wedge integrals of Wigner functions}

\title{Supplementary Information for: All three-angle variants of Tsirelson's precession protocol, and improved bounds for wedge integrals of Wigner functions}

\maketitle

\setcounter{section}{0}
\setcounter{figure}{0}
\setcounter{equation}{0}

\renewcommand{\thesection}{S\arabic{section}}
\renewcommand{\thefigure}{S\arabic{figure}}
\renewcommand{\theequation}{S\arabic{equation}}

\renewcommand{\figurename}{Supplementary Figure}

\section{\label{apd:classical-bound-proof}Proof of Classical Bound}
In this section, we present the proof of the classical bound for the general precession protocol that involves probing at $K$ arbitrary angles, with $K$ odd. Collect the measured angles in the vector $\vec{\theta} := (\theta_k)_{k=1}^{K-1}$ with $0 =: \theta_0 \leq \theta_1 \leq \dots \leq \theta_{K-1} \leq 2\pi$. That is, we place all the angles in a circle and label them in increasing order in a clockwise direction.
\begin{theorem}\label{thm:classical-bound}
The classical bound is $\mathbf{P}^c(\vec{\theta}) = 1 - \delta/K$ for $\delta \in \{0,1,\dots,(K-1)/2\}$ if and only if
\begin{equation}\label{eq:condition-general}
\forall k: \theta_{(k\,\oplus_{_{\!K}}\delta)} \ominus_{_{\!2\pi}} \theta_k \leq \pi.
\end{equation}
\end{theorem}
\begin{proof}
For any initial classical state $(A_x(0),A_y(0))$, the $K$ points will take angles $\phi_0 + \theta_k$ with $\phi_0 := \atantwo(A_y(0),A_x(0))$. Let $\phi_0 + \theta_{k_0}$ be the point in the $A_x >0 \land A_y<0$ plane closest to the $A_y$ axis.

``\emph{if}'': If $(\phi_0 + \theta_{k_0\,\oplus_{_{\!K}}\delta}) \ominus_{_{\!2\pi}} (\phi_0 + \theta_{k_0}) \leq \pi$, this means that there must be at least $1 + \delta$ points on the left of a straight line from  $\phi_0 + \theta_{k_0}$ to the origin, which means that there must be at least $\delta$ points on the negative $A_x$ plane, which therefore implies that there must be at most $K - \delta$ points on the positive $A_x$ plane. Since this is true for every $\phi_0$, Eq.~\eqref{eq:condition-general} implies
\begin{equation}
P^c(\vec{\theta}) \leq \mathbf{P}^c(\vec{\theta}) = \frac{K-\delta}{K} = 1 - \frac{\delta}{K}.
\end{equation}

``\emph{only if}'': Assume Eq.~\eqref{eq:condition-general} is false. Then, there is some $k_0$ such that $\theta_{k_0\,\oplus_{_{\!K}}\delta} \ominus_{_{\!2\pi}} \theta_{k_0} > \pi$. We can choose a $\phi_0$ such that both $\phi_0+\theta_{k_0\,\oplus_{_{\!K}}\delta}$ and $\phi_0+\theta_{k_0}$ are in the positive $A_x$ plane, as the angle from the former to the latter in the anticlockwise direction is less than $\pi$. Since there are $K-\delta-1$ points between the former and latter points excluding those two points, which would all be in the positive $A_x$ plane, the score would be \begin{equation}
P^c(\vec{\theta}) = \frac{(K-\delta-1) + 2}{K}
= \pqty{1 - \frac{\delta}{K}} + \frac{1}{K},
\end{equation}
so $\mathbf{P}^c(\vec{\theta}) \geq P^c(\vec{\theta}) > 1-\delta/K$. Taking the converse, $\mathbf{P}^c(\vec{\theta}) \leq 1-\delta/K$ implies Eq.~\eqref{eq:condition-general}.
\end{proof}

\section{\label{apd:symplectic-relation}Transformation Between Three-angle Generalisations of Tsirelson's Precession Protocol with the Quantum Harmonic Oscillator}

\begin{theorem}\label{prop:boundary} $\mathbf{P}^\infty(\vec{\theta}) = 2/3$ on the boundary of the $\mathbf{P}^c(\vec{\theta}) = 2/3$ region.
\end{theorem}
\begin{proof}
    On the boundary, exactly one or two of Eq.~\eqref{eq:condition} from the main text is satisfied: so, at least one inequality will be saturated and at least one other will be strict. By Eq.~\eqref{eq:angles-equivalence} from the main text, all such boundaries are equivalent. Consider the one defined by $\theta_1 = \pi$ and $\theta_2-\theta_1 < \pi$. Then,
    \begin{equation}
    \begin{aligned}
        Q^\infty(\vec{\theta})
        &= \frac{1}{3}\Bqty\Big{
            \underbrace{\Theta(X) + \Theta(-X)}_{\mathbbm{1}} + \Theta[X(\theta_2)]
        } \\
        &= \frac{1}{3}\mathbbm{1} + \frac{1}{3}\Theta[X(\theta_2)].
        \end{aligned}
    \end{equation}
    Since $\langle \Theta[X(\theta_2)] \rangle \leq 1$, this gives $\mathbf{P}^\infty \leq 2/3$, which is saturated by any state with only positive support on $X(\theta_2)$.
\end{proof}
\begin{theorem}\label{prop:interior}
    $\mathbf{P}^\infty(\vec{\theta}) = \mathbf{P}_3^\infty$ in the interior region $\triangle$ where $\mathbf{P}^c(\vec{\theta}) = 2/3$.
\end{theorem}
\begin{proof}
    The constructive proof involves showing that $Q^\infty(\vec{\theta}) = U(\vec{\theta}) Q^\infty_3 U^\dag(\vec{\theta})$ for a symplectic unitary $U(\vec{\theta})$, which leaves the maximum eigenvalue unchanged. Explicitly, $U(\vec{\theta}) = R_{\phi_0}S_{\lambda_1}R_{\phi_2}S_{\lambda_3}$ is a sequence of squeeze operators $S_\lambda := \exp[-i\lambda(XP+PX)/2]$ and rotation operators $R_\phi := \exp[-i\phi (X^2 + P^2)/2]$, which act on a general quadrature $X(\theta) = X \cos\theta + P \sin\theta$ as
    \begin{equation}
    \begin{aligned}
        S_\lambda^\dag X(\theta) S_\lambda &= X e^{\lambda}  \cos\theta + P e^{-\lambda} \sin\theta, \\
        R_\phi^\dag X(\theta) R_\phi &= X \cos(\theta-\phi) + P\sin(\theta-\phi).
    \end{aligned}
    \end{equation}
    We shall break it down into several steps. First, defining
    \begin{equation}
    \begin{aligned}
        \phi_0 := \begin{cases}
            0 & \begin{array}{l}
            \text{if }
            (\pi < \theta_2 \leq 2\theta_1)
            \land
            (\frac{\pi}{2} < \theta_1 \leq \frac{2\pi}{3}) \\
            \text{or }
            (\pi < \theta_2 \leq \frac{\theta_1}{2} + \pi)
            \land
            (\frac{2\pi}{3} \leq \theta_1 < \pi)
            \end{array} \\[3.5ex]
            \theta_1 & \begin{array}{l}
            \text{if }
            (2\pi - \theta_1 \leq \theta_2 < \pi + \theta_1)
            \land
            (\frac{\pi}{2} < \theta_1 \leq \frac{2\pi}{3}) \\
            \text{or }
            (\frac{\theta_1}{2} + \pi \leq \theta_2 < \pi + \theta_1)
            \land
            (\frac{2\pi}{3} \leq \theta_1 < \pi)
            \end{array} \\[3.5ex]
            \theta_2 & \begin{array}{l}
            \text{if }
            (2\theta_1 \leq \theta_2 \leq 2\pi - \theta_1)
            \land
            (\frac{\pi}{2} < \theta_1 \leq \frac{2\pi}{3}) \\
            \text{or }
            (\pi < \theta_2 < \pi + \theta_1)
            \land
            (0 < \theta_1 \leq \frac{\pi}{2})
            \end{array}
        \end{cases}
    \end{aligned}
    \end{equation}
    we have $R_{\phi_0}^\dag Q^\infty(\vec{\theta}) R_{\phi_0} = Q^\infty(\theta_1',\theta_2')$, where
    \begin{equation}
        \theta_1' = \begin{cases}
            \theta_1 & \text{if $\phi_0 = 0$,}\\
            \theta_2-\theta_1 & \text{if $\phi_0 = \theta_1$,}\\
            2\pi - \theta_2 & \text{if $\phi_0 = \theta_2$;}
        \end{cases}
    \end{equation}
    and
    \begin{equation}
        \theta_2' = \begin{cases}
            \theta_1 & \text{if $\phi_0 = 0$,}\\
            2\pi - \theta_1 & \text{if $\phi_0 = \theta_1$,}\\
            2\pi - (\theta_2-\theta_1) & \text{if $\phi_0 = \theta_2$.}
        \end{cases}
    \end{equation}
    This transformation essentially uses the symmetry in Eq.~\eqref{eq:angles-equivalence} from the main text to bring every point in the interior of the $\mathbf{P}_3^c = 2/3$ region onto the bottom right trine of the purple triangle in Fig.~\ref{fig:changeOfVariable}(b) of the main text, for which $(\pi/2 < \theta_1' < \pi) \land (\pi  < \theta_2' < 3\pi/2)$. Then, defining
    \begin{equation}\label{eq:squeeze-step-1}
        \lambda_1 := \frac{1}{4}\log(\frac{\tan^2\theta_1'\tan\theta_2'}{\tan\theta_2'-2\tan\theta_1'}),
    \end{equation}
    which is well-defined because $\tan\theta_2 > 0$ and $\tan\theta_1 < 0$ in this region, so the argument in the logarithm is strictly positive. This gives
    \begin{widetext}
    \begin{equation}
    \begin{aligned}
        S_{\lambda_1}^\dag Q^\infty(\theta_1',\theta_2') S_{\lambda_1}  &=
        \frac{1}{3}\bqty\Big{
            \Theta\pqty{X e^{\lambda_1}} +
        \Theta\pqty{Xe^{\lambda_1}\cos\theta_1' + P e^{-\lambda_1}\sin\theta_1'} +
            \Theta\pqty{X e^{\lambda_1}\cos\theta_2' + P e^{-\lambda_1}\sin\theta_2'}
        } \\
    &=: \frac{1}{3}\bqty\Big{
            \Theta\pqty{X} +
            \Theta\pqty{X\cos\theta_1'' + P\sin\theta_1''} +
            \Theta\pqty{X \cos\theta_2'' + P\sin\theta_2''}
        } \\
    &= Q^{\infty}(\theta_1'',\theta_2''),
    \end{aligned}
    \end{equation}
    \end{widetext}
    where we used $\Theta(x) = \Theta(cx)$ for any $c > 0$ and defined $\theta_k'' := \atantwo(e^{-\lambda_1}\sin\theta_k',e^{\lambda_1}\cos\theta_k')\bmod 2\pi$ for $k=1,2$. Substituting Eq.~\eqref{eq:squeeze-step-1} into the definition of $\theta_k''$,
    \begin{subequations}
    \begin{align}
        \theta_1'' &= \atantwo\pqty{
            e^{-\lambda_1}\sin\theta_1',
            e^{\lambda_1}\cos\theta_1'
        }\bmod 2\pi \nonumber\\
        &= \pi - \atan\sqrt{1 + \frac{2\abs{\tan\theta_1'}}{\tan\theta_2'}} \\
        &= \frac{\pi}{2} + \atan\frac{1}{\sqrt{1 + 2\abs{\tan\theta_1'}/\tan\theta_2'}} \nonumber\\
        &=:\phi_2, \nonumber\\
        \theta_2'' &= \atantwo\pqty{
            e^{-\lambda_1}\sin\theta_2',
            e^{\lambda_1}\cos\theta_2'
        }\bmod 2\pi \nonumber\\
        &= \pi + \atan\frac{\sqrt{1 + 2\abs{\tan\theta_1'}/\tan\theta_2'}}{\abs{\tan\theta_1'}/\tan\theta_2'} \\
        &=
            2\pqty{\frac{\pi}{2} + \atan\frac{1}{\sqrt{1 + 2\abs{\tan\theta_1'}/\tan\theta_2'}}}\nonumber\\
        &=2\phi_2,\nonumber
    \end{align}
    \end{subequations}
    where we have used $\atan(x) + \atan(1/x) = \pi/2$ and $2\atan(1/\sqrt{2x+1}) = \atan(\sqrt{2x+1}/x)$ for $x > 0$; the latter identity can be proven by taking the tangent on both sides and using the double angle formula. This therefore implies that $S_{\lambda_1}^\dag Q^{\infty}(\theta_1',\theta_2') S_{\lambda_1} = Q^{\infty}(\phi_2,2\phi_2)$, where $\pi/2 < \phi_2 < 3\pi/4$ as $\abs{\tan\theta_1'}/\tan\theta_2'$ is finite and strictly positive in this region.

    Finally, by defining $4\lambda_3 := 2\log|\tan\phi_2| - \log 3$,
    \begin{equation}
    \begin{aligned}
        & S_{\lambda_3}^\dag R_{\phi_2}^\dag Q^{\infty} (\phi_2,2\phi_2) R_{\phi_2}S_{\lambda_3} \\&=
        \frac{1}{3}\Big[
            \Theta\pqty{X e^{\lambda_3} \cos\phi_2 - P e^{-\lambda_3} \sin\phi_2} +
        \Theta\pqty{X} \\
        &\qquad\qquad{}+{}
            \Theta\pqty{X e^{\lambda_3} \cos\phi_2 + P e^{-\lambda_3} \sin\phi_2}
        \Big] \\
        &= Q^{\infty}(\tilde\theta_+,\tilde\theta_-),
    \end{aligned}
    \end{equation}
    where similar to before we have defined $\tilde\theta_\pm := {\atantwo}( \pm e^{-\lambda_3} \sin\phi_2, e^{\lambda_3} \cos\phi_2 ) \bmod 2\pi$. Since $\pi/2 < \phi_2 < 3\pi/4 \implies \tan\phi_2 = -|\tan\phi_2|$,
    \begin{equation}
    \begin{aligned}
        \tilde\theta_{\pm} &= \pi \mp \atan\!\abs{e^{-2\lambda_3} \tan\theta } \\
        &= \pi \mp \atan\sqrt{3} \\
        &= \begin{cases}
        2\pi/3 & \text{for $\tilde\theta_+$,}\\
        4\pi/3 & \text{for $\tilde\theta_-$.}
        \end{cases}
    \end{aligned}
    \end{equation}
    Hence, $S_{\lambda_3}^\dag R_{\phi_2}^\dag Q^{\infty} (\phi_2,2\phi_2) R_{\phi_2}S_{\lambda_3} = Q^{\infty}(2\pi/3,4\pi/3) = Q^{\infty}_3$. Putting everything together,
    \begin{equation}
        Q^{\infty}(\vec\theta) = \underbrace{R_{\phi_0}S_{\lambda_1}R_{\phi_2}S_{\lambda_3}}_{=: U(\vec{\theta})} Q^{\infty}_3 S_{\lambda_3}^\dag R_{\phi_2}^\dag S_{\lambda_1}^\dag R_{\phi_0}^\dag,
    \end{equation}
    with $\phi_0$, $\lambda_1$, $\phi_2$, and $\lambda_3$ as defined above. Since this shows that $Q_3^{\infty}$ and $Q^{\infty}(\vec\theta)$ are related by the unitary transformation $U(\vec{\theta})$, which leaves the eigenvalues unchanged, this implies that $\mathbf{P}^\infty(\vec{\theta}) = \mathbf{P}_3^\infty$, which completes the proof. Note that this also means that if the state $\ket{P_3}$ obtains the score $P_3$ for the original protocol, then the state $U(\vec{\theta})\ket{P_3}$ obtains the same score for the protocol with probing angles $\vec{\theta}$.
\end{proof}

\section{\label{apd:lower-bound-expression}Lower Bound Expressions}
The matrix element of $\Theta(X)-\mathbbm{1}/2$ is known to be \cite{tsirelson-spin}
\begin{widetext}
\begin{equation}
    \bra{n}\pqty{\Theta(X)-\frac{1}{2}\mathbbm{1}}\ket{n'} = \begin{cases}
    0 & \text{if $n\ominus_{2}n'=0$,} \\
    \displaystyle\frac{(-1)^{\frac{n-n'-1}{2}} 2^{-\frac{n+n'}{2}}}{n-n'}
    \sqrt{
        \frac{n^{(n\bmod 2)}n^{\prime (n'\bmod 2)}}{\pi}
        \pmqty{
            2\lfloor \frac{n}{2} \rfloor\\
            \frac{n}{2}
        }
        \pmqty{
            2\lfloor\frac{n'}{2}\rfloor\\
            \lfloor\frac{n'}{2}\rfloor
        }
    } & \text{otherwise,}
    \end{cases}
\end{equation}
while $Q_3^\infty-\mathbbm{1}/2$ is given by the direct sum
\begin{equation}
\begin{aligned}
    Q_3^\infty-\frac{1}{2}\mathbbm{1} &= \bigoplus_{k=0}^2 \bqty{\Pi_{k|3}^{-}\pqty{\Theta(X)-\frac{1}{2}\mathbbm{1}}\Pi_{k|3}^{+} + \Pi_{k|3}^{+}\pqty{\Theta(X)-\frac{1}{2}\mathbbm{1}}\Pi_{k|3}^{-}},
\end{aligned}
\end{equation}
where $\Pi_{k|3}^{+} := \sum_{n=0}^\infty \ketbra{6n+k}{6n+k}$ and  $\Pi_{k|3}^{-} := \sum_{n=0}^\infty \ketbra{6n+3+k}{6n+3+k}$. Since $A_3 = (Q_3^\infty-\frac{1}{2}\mathbbm{1})^2 - (\mathbf{P}_3^c-\frac{1}{2})^2\mathbbm{1}$,
\begin{equation}\label{eq:Q3-square-mel-definition}
    \bra{6n+k}A_3\ket{6n'+k} = \sum_{m=0}^\infty \bra{6n+k}\pqty{\Theta(X)-\frac{1}{2}\mathbbm{1}}\ketbra{6m + k + 3}{6m + k + 3}\pqty{\Theta(X)-\frac{1}{2}\mathbbm{1}}\ket{6n'+k}
    - \delta_{n,n'}\pqty{\mathbf{P}_3^c-\frac{1}{2}}^2
\end{equation}
with $k \in \{0,1,2\}$. Let us first consider $k$ even. Then, we can write Eq.~\eqref{eq:Q3-square-mel-definition} as the sum $\sum_{m=0}^\infty c^{n,n'}_{3m+(1+k/2)} = \sum_{m\bmod 3 = 1+k/2}^\infty c^{n,n'}_{m}$, where
\begin{equation}
c^{n,n'}_{m} := \bra{6n+k}\pqty{\Theta(X)-\frac{1}{2}\mathbbm{1}}\ketbra{2m +1}{2m + 1}\pqty{\Theta(X)-\frac{1}{2}\mathbbm{1}}\ket{6n'+k}.
\end{equation}
In other words, Eq.~\eqref{eq:Q3-square-mel-definition} is a subseries of the series $\{c^{n,n'}_m\}_{m=0}^\infty$, so the matrix element can be rewritten as \cite{subseries}
\begin{equation}
    \bra{6n+k}A_3\ket{6n'+k} = \frac{1}{3}\sum_{j=0}^2 e^{-i\frac{2\pi j}{3}(1+\frac{k}{2})} \bqty{\sum_{m=0}^\infty c_m^{3n+\frac{k}{2},3n'+\frac{k}{2}} \pqty{e^{i\frac{2 \pi j}{3}}}^m} - \delta_{n,n'}\pqty{\mathbf{P}_3^c-\frac{1}{2}}^2.
\end{equation}
The full series in the square bracket can be interpreted as the Maclaurin series of a function with the argument $e^{i(2\pi j/3)}$. Isolating the parts of the function where the series needs to be resolved, we have
\begin{equation}
\begin{aligned}
    \sum_{m=0}^\infty c_m^{n,n'} z^m &=  \frac{(-1)^{n+n'}2^{-(n+n')}}{8\pi}\sqrt{
        \binom{2n}{n}\binom{2n'}{n'}} \underbrace{\sum_{m=0}^\infty \binom{2m}{m} \frac{2(m+\frac{1}{2})2^{-2m}}{((m-n)+\frac{1}{2})((m-n')+\frac{1}{2})}  z^m}_{:=l(z;n,n')}.
\end{aligned}
\end{equation}
Finally, by using the series definitions of the generalised hypergeometric function $_3F_2(a_1,a_2,a_3;b_1,b_2;z) = \sum_{n=0}^\infty \frac{(a_1+n-1)!(a_2+n-1)!(a_3+n-1)!(b_1-1)!(b_2-1)!}{(b_1+n-1)!(b_2+n-1)!(a_1-1)!(a_2-1)!(a_3-1)!}\frac{z^n}{n!}$ and the incomplete beta function $B(z;a,b) = \int_0^z t^{a-1}(1-t)^{b-1} \dd{t} = \sum_{n=0}^\infty \frac{(n-b)!z^{n+a}}{(-b)!n!(n+a)}$, we obtain
\begin{equation}
l(z;n,n') = \begin{cases}
    \frac{ {_3F_2}\left(\frac{1}{2},\frac{1}{2}-n,\frac{1}{2}-n;\frac{3}{2}-n,\frac{3}{2}-n;z\right)}{(\frac{1}{2}-n)^2}+ \frac{z \, {_3F_2}\left(\frac{3}{2},\frac{3}{2}-n,\frac{3}{2}-n;\frac{5}{2}-n,\frac{5}{2}-n;z\right)}{(\frac{3}{2}-n)^2} & \text{if $n = n'$,} \\
    \begin{array}{l}
        2z^{n'-1/2}B(z;\frac{1}{2}-n',\frac{1}{2})
    \end{array} & \text{if $0 = n \neq n'$,}\\
    \begin{array}{l}
        2z^{n-1/2} B(z;\frac{1}{2}-n,\frac{1}{2})
    \end{array} & \text{if $n \neq n' = 0$,}\\
    \begin{array}{l}
        \frac{z^{n-1/2}}{n-n'}B(z;\frac{1}{2}-n,-\frac{1}{2}) + \frac{z^{n'-1/2}}{n'-n}B(z;\frac{1}{2}-n',-\frac{1}{2})
    \end{array} & \text{otherwise.}
\end{cases}
\end{equation}
With this, for $k$ even, the matrix elements of $A_3$ is
\begin{equation}
\begin{aligned}
    \bra{6n+k}A_3\ket{6n'+k} &= \frac{(-1)^{(n+n')}2^{-(3(n+n')+k)}}{24\pi}\sqrt{\binom{6n+k}{3n+\frac{k}{2}}\binom{6n'+k}{3n'+\frac{k}{2}}}\sum_{j=0}^2 e^{-i\frac{2\pi j}{3}(1+\frac{k}{2})} l(e^{i\frac{2\pi j}{3}};3n+\tfrac{k}{2},3n'+\tfrac{k}{2})\\
    &\qquad\qquad {}-{} \delta_{n,n'}\pqty{\mathbf{P}_3^c-\frac{1}{2}}^2,
\end{aligned}
\end{equation}
and similar steps for $k$ odd give
\begin{equation}
\begin{aligned}
    \bra{6n+3+k}A_3\ket{6n'+3+k} &= \frac{(-1)^{(n+n')}2^{-(3(n+n'+1)+k)}}{24\pi}\sqrt{\binom{6n+3+k}{3n+\frac{3+k}{2}}\binom{6n'+3+k}{3n'+\frac{3+k}{2}}}\\
    &\qquad\qquad\qquad{}\times{}\sum_{j=0}^2 e^{-i\frac{2\pi j}{3}(\frac{k-1}{2})} l(e^{i\frac{2\pi j}{3}};3n+\tfrac{3+k}{2},3n'+\tfrac{3+k}{2}) - \delta_{n,n'}\pqty{\mathbf{P}_3^c-\frac{1}{2}}^2.
\end{aligned}
\end{equation}
\end{widetext}

\section{\label{apd:upper-bound-expression}Upper Bound Expressions}
In the section titled \emph{Improved Rigorous Bounds} in the main text, the inequality $2(\mathbf{P}_3^\infty-\frac{2}{3})^2(\mathbf{P}_3^\infty-\frac{1}{3})^2 \leq \tr(A_3^2)$ was found by using the double degeneracies of the nonzero eigenvalues of $A_3$ and the relationship between the eigenvalues of $Q_3^\infty$ and $A_3$. In this section, we detail the steps required to obtain the exact value of the latter expression.

To do so, we require the formalism of Wigner functions. For a continuous variable system with a single degree of freedom, recall that the Wigner function of a state $\rho$ is defined as \cite{wigner-review}
\begin{equation}
    W_\rho(x,p) := \frac{1}{2\pi i}\tr(\rho e^{i\frac{\pi}{2}[(X-x)^2+(P-p)^2)]} ),
\end{equation}
while the Wigner function of an observable $A$ is the function $W_A(x,p)$ that satisfies $\tr(A\rho) = \int\dd{x}\int\dd{p} W_A(x,p) W_\rho(x,p)$ for every state $\rho$, which may be defined in the sense of a distribution. The Wigner function of $\Theta[X(\theta)]\circ\Theta[X(\theta')]:=\frac{1}{2}(\Theta[X(\theta)]\Theta[X(\theta')] +\Theta[X(\theta')]\Theta[X(\theta)])$, where $X(\theta) = X  \cos\theta + P\sin\theta$, was derived by Tsirelson to be \cite{tsirelson-og}
\begin{equation}\label{eq:heaviside-Wigner}
\begin{aligned}
    W_{\Theta[X(\theta)]\circ\Theta[X(\theta')]}(x,p) &= \frac{1}{2}\pqty\big{\Theta[x(\theta)]+\Theta[x(\theta')]} \\
    &\qquad{}+{} \frac{1}{2\pi}\si\!\bqty{\frac{2 x(\theta)x(\theta')}{\abs{\sin(\theta-\theta')}}},
\end{aligned}
\end{equation}
where $\si(x) = -\int_{x}^\infty\dd{t} \frac{\sin t}{t}$ and $x(\theta) = x\cos\theta + p\sin\theta$. A few special values of $\si(x)$ are $\si(-\infty) = \pi$, $\si(0) = \pi/2$, and $\si(\infty) = 0$, which come from standard integrals of $\sin(x)/x$. The last special value also implies that $W_{(\Theta[X(\theta)])^2}(x,p) = \Theta[x(\theta)]$, which is the expected result coming from $(\Theta[X(\theta)])^2 = \Theta[X(\theta)]$

Using Eq.~\eqref{eq:heaviside-Wigner}, the Wigner function of $A_3 = (Q_3^\infty-\frac{2}{3}\mathbbm{1})(Q_3^\infty-\frac{1}{3}\mathbbm{1})$ can be found to be
\begin{equation}
    W_{A_3}(x,p) = \frac{2}{9} + \frac{1}{18}\sum_{k=0}^2\frac{2}{\pi}\si\!\bqty{\frac{4}{\sqrt{3}}x\pqty{\frac{2\pi k}{3}}x\pqty{\frac{2\pi (k+1)}{3}}}.
\end{equation}
In terms of their Wigner functions, the trace of two observables is given by
\begin{equation}
    \tr(AB) = \frac{1}{2\pi}\int\dd{x}\int\dd{p} W_A(x,p)W_B(x,p),
\end{equation}
with which we calculate ${\tr}(A_3^2)$ to be
\begin{widetext}
\begin{equation}\label{eqA:trace-integral}
\begin{aligned}
    \tr(A_3^2) &= \frac{1}{2\pi}\int_{-\infty}^\infty\dd{x}\int_{-\infty}^\infty\dd{p}\bqty{W_{A_3}(x,p)}^2 \\
    &= \frac{1}{2\pi}\int_{-\infty}^{\infty}\dd{x}\int_{-\infty}^{\infty}\dd{p} \Bqty{\frac{2}{9} + \frac{1}{18}\sum_{k=0}^2\frac{2}{\pi}\si\!\bqty{\frac{4}{\sqrt{3}}x\pqty{\frac{2\pi k}{3}}x\pqty{\frac{2\pi (k+1)}{3}}}}^2 \\
    &= \frac{1}{162\pi^3}\int_{-\infty}^{\infty}\dd{x}\int_{-\infty}^{\infty}\dd{p} \Bqty{\sum_{k=0}^2(-1)^{\delta_{k,1}}\si\!\bqty{-(-1)^{\delta_{k,1}}\frac{4}{\sqrt{3}}x\pqty{\frac{2\pi k}{3}} x\pqty{\frac{2\pi (k+1)}{3}}}}^2 \\
    &= \frac{1}{324\pi^3}\int_{-\pi}^{\pi}\dd{\phi} \int_{0}^\infty\dd{[r^2]} \Bqty{
    \si\!\bqty{\frac{4r^2}{\sqrt{3}}\cos(\phi+\frac{\pi}{3})\cos(\phi-\frac{\pi}{3})}
    - \sum_{\sigma\in \{+,-\}}\si\!\bqty{\frac{4r^2}{\sqrt{3}}\cos(\phi)\cos(\phi+\sigma\frac{\pi}{3})}
    }^2 \\
    &= \frac{1}{27\pi^3}\int_{0}^{\frac{\pi}{6}}\dd{\phi} \int_{0}^\infty\dd{[r^2]} \Bqty{
    \si\!\bqty{\frac{4r^2}{\sqrt{3}}\cos(\phi+\frac{\pi}{3})\cos(\phi-\frac{\pi}{3})}
    - \sum_{\sigma\in \{+,-\}}\si\!\bqty{\frac{4r^2}{\sqrt{3}}\cos(\phi)\cos(\phi+\sigma\frac{\pi}{3})}
    }^2,
\end{aligned}
\end{equation}
where we have used $\si(-x) = -\si(x) - \si(-\infty) =  -\si(x) -\pi$ in the second line. In the penultimate line, a change of variables to radial coordinates $x=r\cos\phi$ and $p=r\sin\phi$ was performed, and we also used the identity $-\cos(x) = \cos(x-\pi)$; in the last line, we used the fact that the integrand is symmetric under the transformations $\phi \to \phi + \pi/3$ and $\phi \to -\phi$, which implies that the integral over $-\pi < \phi < \pi$ is simply $12$ times of the integral over $0<\phi<\pi/6$.

The integration over the radial variable involve integrals of the form $\int_{0}^{\infty}\dd{x}\si(ax)\si(bx)$ for some $0 < a \leq b < \infty$. We can evaluate the integral $\int_{0}^{L}\dd{x}\si(ax)\si(bx)$, where we will later take $L \to \infty$, to be
\begin{equation}\label{eq:integrated-si-si}
\begin{aligned}
    \int_{0}^{L}\dd{x}\si(ax)\si(bx) = \frac{\pi}{2b}
    + \si(aL)\si(bL)L
    - \frac{b-a}{2ab}\si[(b-a)L]
    - \frac{b+a}{2ab}\si[(b+a)L]
    + \frac{\cos(aL)}{a}\si(bL)
    + \frac{\cos(bL)}{b}\si(aL),
\end{aligned}
\end{equation}
\end{widetext}
which can be verified by taking the derivative on both sides with respect to $L$ using $\frac{d}{dL} \si(cL) =\sin(cL)/L$ and also confirming that both sides of the equation are zero when $L=0$ using $\si(0) = \pi/2$.

Now, to take the limit $L\to \infty$, we first take care of the troublesome term $\si(aL)\si(bL)L$ using integration by parts to find that
\begin{equation}
\begin{aligned}
    \sqrt{x}\si(x) &= \sqrt{x}\int_x^\infty\dd{t} \frac{1}{t}\frac{d\cos(t)}{dt} \\
    &= \sqrt{x}\bqty{\frac{\cos(t)}{t}}_x^\infty + \int_x^\infty\dd{t} \frac{\sqrt{x}\cos(t)}{t^2} \\
    &= -\frac{\cos(x)}{\sqrt{x}} + \int_x^\infty\dd{t} \frac{\sqrt{x}\cos(t)}{t^2},
\end{aligned}
\end{equation}
whose magnitude for $x > 0$ is therefore bounded as
\begin{equation}
\begin{aligned}
\abs{\sqrt{x}\si(x)} &\leq \frac{\abs{\cos(x)}}{\sqrt{x}} + \int_x^\infty\dd{t} \frac{\sqrt{x}\abs{\cos(t)}}{t^2} \\
&\leq \frac{1}{\sqrt{x}} + \int_x^\infty\dd{t} \frac{1}{t^{3/2}} = \frac{3}{\sqrt{x}}.\\[6ex]
\end{aligned}
\end{equation}
This implies that $\lim_{x\to\infty}\abs{\sqrt{x}\si(x)} \leq \lim_{x\to\infty} 3/\sqrt{x} = 0$, and thus $\lim_{L\to\infty}\si(aL)\si(bL)L = 0$. Every other term in Eq.~\eqref{eq:integrated-si-si} that depend on $L$ vanish due to $\lim_{L\to\infty}\si(cL) = \si(\infty) = 0$ when $c>0$, which leaves only the first term $\pi/(2b)$. Therefore, $\int_{0}^{\infty}\dd{x}\si(ax)\si(bx) = \pi/(2b)$ for any $0<a \leq b<\infty$.

To apply this to the integral over $r^2$ in Eq.~\eqref{eqA:trace-integral}, we use the monotonicity of $\cos(x)$ within the range $0\leq x \leq \pi$ to ascertain that $\cos(\phi+\frac{\pi}{3}) \leq \cos(\phi-\frac{\pi}{3}) \leq \cos(\phi)$ for $0<\phi<\pi/6$, with which we can determine that
\begin{equation}
\begin{aligned}
    \frac{4}{\sqrt{3}}\cos(\phi+\tfrac{\pi}{3})\cos(\phi-\tfrac{\pi}{3}) &\leq \frac{4}{\sqrt{3}}\cos(\phi)\cos(\phi+\tfrac{\pi}{3}) \\
    &\leq \frac{4}{\sqrt{3}}\cos(\phi)\cos(\phi-\tfrac{\pi}{3})
\end{aligned}
\end{equation}
within the range of $\phi$ we will be integrating over. Finally, by making liberal use of the identities $\cos(a \pm b) = \cos(a)\cos(b) \mp \sin(a)\sin(b)$ and $\sin(a \pm b) = \sin(a)\cos(b) \pm \cos(a)\sin(b)$, and the special values $\cos(\pi/3) = 1/2$ and $\sin(\pi/3) = \sqrt{3}/2$, we end up with
\begin{widetext}
\begin{equation}\label{eq:upper-bound-angle-integral}
\begin{aligned}
    \tr(A_3^2)
    &= \frac{1}{27\pi^3}\int_{0}^{\frac{\pi}{6}}\dd{\phi} \Bigg\{
        \frac{\pi/2}{\frac{4}{\sqrt{3}}\cos(\phi+\tfrac{\pi}{3})\cos(\phi-\tfrac{\pi}{3})} +
        \frac{\pi/2}{\frac{4}{\sqrt{3}}\cos(\phi)\cos(\phi+\tfrac{\pi}{3})} +
        \frac{\pi/2}{\frac{4}{\sqrt{3}}\cos(\phi)\cos(\phi-\tfrac{\pi}{3})} \\
    &\qquad\qquad\qquad\qquad\quad{}-{}
        \frac{\pi}{\frac{4}{\sqrt{3}}\cos(\phi)\cos(\phi+\tfrac{\pi}{3})} -
        \frac{\pi}{\frac{4}{\sqrt{3}}\cos(\phi)\cos(\phi-\tfrac{\pi}{3})} +
        \frac{\pi}{\frac{4}{\sqrt{3}}\cos(\phi)\cos(\phi-\tfrac{\pi}{3})}
    \Bigg\}\\
    &= \frac{1}{72\sqrt{3}\pi^2}\int_{0}^{\frac{\pi}{6}}\dd{\phi} \frac{\cos(\phi)-\cos(\phi-\frac{\pi}{3})+\cos(\phi+\frac{\pi}{3})}{\cos(\phi)\cos(\phi+\frac{\pi}{3})\cos(\phi-\frac{\pi}{3})}\\
    &= \frac{1}{72\sqrt{3}\pi^2}\int_{0}^{\frac{\pi}{6}}\dd{\phi} \frac{\cos(\phi)-2\sin(\frac{\pi}{3})\sin(\phi)}{\cos(\phi)\cos(\phi+\frac{\pi}{3})\cos(\phi-\frac{\pi}{3})}\\
    &= \frac{1}{36\sqrt{3}\pi^2}\int_{0}^{\frac{\pi}{6}}\dd{\phi} \frac{\frac{1}{2}\cos(\phi)-\frac{\sqrt{3}}{2}\sin(\phi)}{\cos(\phi)\cos(\phi+\frac{\pi}{3})\cos(\phi-\frac{\pi}{3})}\\
    &= \frac{1}{36\sqrt{3}\pi^2}\int_{0}^{\frac{\pi}{6}}\dd{\phi} \frac{1}{\cos(\phi)\cos(\phi-\frac{\pi}{3})}\\
    &= \frac{1}{54\pi^2}\int_{0}^{\frac{\pi}{6}}\dd{\phi} \bqty{\tan(\phi+\frac{\pi}{6}) - \tan(\phi-\frac{\pi}{6})} \\
    &= \frac{6\log 2}{(18\pi)^2}.
\end{aligned}
\end{equation}
\end{widetext}

\section{\label{apd:separable-bound}Separable Bounds of Coupled Harmonic Oscillators}
\begin{theorem}
    For all $\vec{\theta}$ in the region where $\mathbf{P}^c(\vec{\theta}) = 2/3$, $\mathbf{P}^{\infty\mathrm{\text{-}sep}}(\varphi,\vec{\theta}) = \mathbf{P}^{\infty\mathrm{\text{-}sep}}_3(\phi)$.
\end{theorem}
\begin{proof}
    For all points in the interior region, $Q_{\sigma}^\infty(\varphi,\vec{\theta}) = U_\sigma^\dag(\vec{\theta}) Q_{3,\sigma}^\infty(\varphi) U_\sigma(\vec{\theta})$ for some symplectic unitary $U_\sigma(\vec{\theta})$ defined on the collective coordinate $X_{\sigma\varphi}$, due to Theorem~\ref{prop:interior}. Define $U_{\pm\sigma}(\vec{\theta})$ as
    \begin{equation}
         U_{\pm\sigma}(\vec{\theta}) = \exp[i\pmqty{X_{\pm\sigma\varphi}\\P_{\pm\sigma\varphi}}^\dag\Lambda\pmqty{X_{\pm\sigma\varphi}\\P_{\pm\sigma\varphi}}],
    \end{equation}
    where $\Lambda = \Lambda^T$ specifies the symplectic transformation $U_{\sigma}(\vec{\theta})$, and $U_{-\sigma}(\vec{\theta})$ is defined with the same $\Lambda$ but on the coordinate $X_{-\sigma\varphi}$. Notice also that we can relate the local and global coordinates as
    \begin{equation}
        \pmqty{
            X_{+\varphi}\\
            X_{-\varphi}\\
            P_{+\varphi}\\
            P_{-\varphi}
        } = \bqty\bigg{\mathbbm{1}\otimes \underbrace{\pmqty{
                \cos\varphi & \sin\varphi\\
                -\sin\varphi & \cos\varphi
            }}_{\mathcal{R}_\varphi}}\pmqty{
            X_1\\
            X_2\\
            P_1\\
            P_2
        },
    \end{equation}
    and thus
    \begin{equation}
    \begin{aligned}
        &U_+(\vec{\theta})U_-(\vec{\theta}) \\
        &= \exp[i\pmqty{
            X_{+\varphi}\\
            X_{-\varphi}\\
            P_{+\varphi}\\
            P_{-\varphi}
        } ^\dag\pqty{\Lambda\otimes\mathbbm{1}}\pmqty{
            X_{+\varphi}\\
            X_{-\varphi}\\
            P_{+\varphi}\\
            P_{-\varphi}
        } ] \\
        &=  \exp[i\pmqty{
            X_{1}\\
            X_{2}\\
            P_{1}\\
            P_{2}
        }^\dag
        \pqty{\mathbbm{1}\otimes\mathcal{R}_\varphi^\dag}
        \pqty{\Lambda\otimes\mathbbm{1}}
        \pqty{\mathbbm{1}\otimes\mathcal{R}_\varphi}
        \pmqty{
            X_{1}\\
            X_{2}\\
            P_{1}\\
            P_{2}
        } ] \\
        &= U_{1}U_{2},
    \end{aligned}
    \end{equation}
    where $U_{n} = {\exp}[i\spmqty{X_{n}\\P_{n}}^\dag\Lambda\spmqty{X_{n}\\P_{n}}]$ are now unitaries defined in the local coordinates. Therefore,
    \begin{equation}
    \begin{aligned}
        Q_{\sigma}^\infty(\varphi,\vec{\theta})
        &= U_\sigma^\dag(\vec{\theta}) U_{-\sigma}^\dag(\vec{\theta}) Q_{3,\sigma}^\infty(\varphi) U_\sigma(\vec{\theta}) U_{-\sigma}(\vec{\theta}) \\
        &= U_1^\dag U_2^\dag Q_{3,\sigma}^\infty(\varphi) U_1 U_2.
    \end{aligned}
    \end{equation}
    Since the maximisation over the set of positive partial transpose states is unchanged under local unitary operations, we have $\mathbf{P}^{\infty\text{-sep}}(\varphi,\vec{\theta}) = \mathbf{P}^{\infty\text{-sep}}_3(\varphi)$ as desired.
\end{proof}

\section{\label{apd:finitestats}Proof of Finite Measurement Precision}
In the section titled \emph{Impact of Limited Measurement Precision} in the main text, we considered the impact of limited measurement precision as might emerge in experimental implementations of the protocol. The discussion relies on the following theorem.

\begin{theorem}\label{thm:finitestats}
    Consider $\vec{\theta}\in\triangle$ such that $\mathbf{P}^c(\vec{\theta}) = \mathbf{P}_3^c$. Define $\widetilde{\mathbf{P}}^c(\vec{\theta})$ to be the maximum classical score obtained when performing the three-angle protocol with the score assignment $\widetilde{\Theta}[A_x(\theta)]$ instead of $\Theta[A_x(\theta)]$, where
    \begin{equation}
        \widetilde{\Theta}[A_x(\theta)] := \begin{cases}
            1 & \text{if $A_x(\theta) > \varepsilon_+$,} \\
            1/2 & \text{if $-\varepsilon_- \leq A_x(\theta) \leq  \varepsilon_+$,} \\
            0 & \text{otherwise,}
        \end{cases}
    \end{equation}
    and $0 \leq \varepsilon_- \leq \varepsilon_+ < \infty$. Then, the classical bound $\widetilde{\mathbf{P}}^c(\vec{\theta}) = \mathbf{P}_3^c$ is unchanged.
\end{theorem}
\begin{proof}
    Let us first show that $\widetilde{\mathbf{P}}^c(\vec{\theta}) \geq \mathbf{P}_3^c$. Because $\mathbf{P}^c(\vec{\theta}) \geq \mathbf{P}_3^c$, there exists a classical state with $A_x(\theta) =: r\cos(\theta-\phi)$ such that
    $\Theta[A_x(\theta_0)] = \Theta[A_x(\theta_1)] = 1$ and $\Theta[A_x(\theta_2)] = 0$. This implies that $r\cos(\theta_0-\phi)>0$, $r\cos(\theta_1-\phi) > 0$, and $r\cos(\theta_2) < 0$ for all $r>0$. Therefore, for any choice of $0 \leq \varepsilon_- \leq \varepsilon_+ < \infty$, we can always find an $r$ large enough such that $r\cos(\theta_0-\phi) > \varepsilon_+$, $r\cos(\theta_1-\phi) > \varepsilon_+$, and $r\cos(\theta_2-\phi) < -\varepsilon_-$.
    For such an $r$, we have $\widetilde{\Theta}[A_x(\theta_0)] = \widetilde{\Theta}[A_x(\theta_1)] = 1$ and $\widetilde{\Theta}[A_x(\theta_2)] = 0$, which therefore implies that $\mathbf{P}_3^c = \sum_{k=0}^3 \widetilde{\Theta}[A_x(\theta_k)]/3 \leq \widetilde{\mathbf{P}}^c(\vec{\theta})$.

    Next, let us show that $\widetilde{\mathbf{P}}^c(\vec{\theta}) \leq \mathbf{P}_3^c$ using a proof by contradiction. First, assume that $\widetilde{\mathbf{P}}^c(\vec{\theta}) > \mathbf{P}_3^c$, which either means that  $\widetilde{\mathbf{P}}^c(\vec{\theta}) = 1$ or $\widetilde{\mathbf{P}}^c(\vec{\theta}) = (1+1+1/2)/3 = 5/6$. The first is not possible, since
    \begin{equation}
    \begin{aligned}
        \widetilde{\mathbf{P}}^c(\vec{\theta}) = 1
        &\implies \forall k : A_x(\theta_k) > \varepsilon_+ \geq 0 \\
        &\implies \forall k : \Theta[A_x(\theta_k)] = 1 \\
        &\implies \mathbf{P}^c(\vec{\theta}) = 1,
    \end{aligned}
    \end{equation}
    which contradicts $\mathbf{P}^c(\vec{\theta}) = \mathbf{P}_3^c$.

    Assuming instead that $\widetilde{\mathbf{P}}^c(\vec{\theta}) = 5/6$, there must exist a classical state such that $\widetilde{\Theta}[A_x(\theta_{k_-})] = 1/2$ for some $k_-$ and $\forall k\neq k_- : \widetilde{\Theta}[A_x(\theta_{k})] = 1 \implies \Theta[A_x(\theta_{k})] = 1$. The requirement that $\mathbf{P}^c(\vec{\theta}) = \mathbf{P}_3^c$ further implies that $\Theta[A_x(\theta_{k_-})] = 0$.

    Without any loss of generality, let $k_- = 2$ up to an offset of the measured angles, and rewrite $A_x(\theta) = r\cos(\theta-\phi)$ in polar coordinates as before. From Theorem~\ref{thm:classical-bound}, we also have
    \begin{equation}\label{eq:three-angle-req}
    \begin{gathered}
        0 \leq \theta_1-\theta_0 \leq \pi,\\
        0 \leq \theta_2-\theta_1 \leq \pi,\\
        \pi \leq \theta_2-\theta_0 \leq 2\pi,
    \end{gathered}
    \end{equation}
    which are required in order to satisfy $\mathbf{P}^c(\vec{\theta}) = \mathbf{P}_3^c$. Now, $0 \leq \theta_1-\theta_0 \leq \pi$ implies that the trajectory $r\cos(\theta-\phi)$ for $\theta_0 \leq \theta \leq \theta_1$ is at most a semicircular path in $A_x$--$A_y$ space, and thus cannot exit and reenter the $A_x > \varepsilon_+$ plane. Since $\widetilde{\Theta}[A_x(\theta_0)] = \widetilde{\Theta}[A_x(\theta_1)] = 1$, this means that the entire trajectory must remain in the $A_x > \varepsilon_+$ plane, and therefore $\forall \theta_0 \leq \theta \leq \theta_1: r\cos(\theta-\phi) > \varepsilon_+$.

    Finally, rearranging the last two lines of Eq.~\eqref{eq:three-angle-req}, we have $\theta_0 \leq \theta_2 - \pi \leq \theta_1$. Therefore,
    \begin{equation}
    \begin{aligned}
    &\begin{aligned}
        \varepsilon_+ &< r\cos(\theta_2-\pi-\phi) \\
        &= -r\cos(\theta_2-\phi) = -A_x(\theta_2)
    \end{aligned} \\
    &\implies A_x(\theta_2) < - \varepsilon_+ \leq -\varepsilon_- \\
    &\implies \widetilde{\Theta}[A_x(\theta_2)] = 0,
    \end{aligned}
    \end{equation}
    which contradicts the assumption that $\widetilde{\Theta}[A_x(\theta_{k_-})] = 1/2 \land \forall k \neq k_- : \widetilde{\Theta}[A_x(\theta_{k})] = 1$, which in turn contradicts $\widetilde{\mathbf{P}}^c(\vec{\theta}) = 5/6$. Taking the contrapositive implications, this means that $\widetilde{\mathbf{P}}^c(\vec{\theta}) \leq \mathbf{P}_3^c$, which together with $\widetilde{\mathbf{P}}^c(\vec{\theta}) \geq \mathbf{P}_3^c$ completes the proof.
\end{proof}

We can recast the above in terms of the coarse-grained measurements specified by the end points $\{a_{n}\}_{n \in \mathbb{Z}}$, where the $n$th bin contains the outcomes $a_{n} \leq A_x(\theta) < a_{n+1}$. We shall label the $0$th bin as the bin such that $a_{0} \leq 0 < a_{1}$, and the $n_+$th bin as the one that satisfies
\begin{equation}\label{eq:bin-epsilon}
a_{n_++1} - a_1 > a_{1}-a_0 \implies a_{\mathbf{n}+1} > 2a_1 - a_0 \geq - a_0.
\end{equation}
Then, if $A_x(\theta)$ falls somewhere within the $0$th to $\mathbf{n}$th bins, it must be that $a_0 \leq A_x(\theta) \leq a_{\mathbf{n}+1}$. Setting $\varepsilon_- := -a_0$ and $\varepsilon_+ := a_{\mathbf{n}+1}$, we have $0 \leq \varepsilon_- \leq \varepsilon_+$ due to Eq.~\eqref{eq:bin-epsilon}, as required by Theorem~\ref{thm:finitestats}. Therefore, by setting
\begin{equation}
    \widetilde{\Theta}[A_x(\theta)] := \begin{cases}
        1 & \text{if $A_x(\theta) \in n$th bin with $n > \mathbf{n}$,} \\
        1/2 & \text{if $A_x(\theta) \in n$th bin with $0 \leq n \leq \mathbf{n}$,} \\
        0 & \text{otherwise,}
    \end{cases}
\end{equation}
like in the main text, Theorem~\ref{thm:finitestats} implies that $\widetilde{\mathbf{P}}^c(\vec{\theta}) = \mathbf{P}_3^c$. We drop the tildes in the main text, where it is understood that the replacement $\Theta[A_x(\theta)] \to \widetilde{\Theta}[A_x(\theta)]$ is performed when implementing the protocol.

\section{\label{apd:upper-bound-general}Rigorous Upper Bound for Precession Protocol With More Angles}
From Eq.~\eqref{eq:heaviside-Wigner}, the Wigner function of $A(\vec{\theta}) := (Q(\vec{\theta})-\frac{K+1}{2K})(Q(\vec{\theta})-\frac{K-1}{2K})$ for $K Q(\vec{\theta}) := \sum_{k=0}^{K-1}\Theta[X(\theta_k)]$ is
\begin{equation}
    W_{A(\vec{\theta})}(x,p) = \frac{K^2-1}{4K^2} +\frac{1}{2\pi K^2}\sum_{j=0}^{K-1}\sum_{k=0}^{K-1} \si\bqty{ \frac{2 x(\theta_j) x(\theta_k)}{\abs{\sin(\theta_j-\theta_k)}} }.
\end{equation}
For the region where $\mathbf{P}^c(\vec{\theta} = (1+1/K)/2$, $(K-1)/2$ of the angles are always on one side and $(K+1)/2$ of the angles on the other, so $(K-1)/2\cdot (K-1)/2 + (K-1)/2 \cdot (K-1)/2 = (K^2-1)/2$ of the terms $x(\theta_j)x(\theta_k)$ will be negative, for which
\begin{equation}
    \si\bqty{ \frac{2 x(\theta_j) x(\theta_k)}{\abs{\sin(\theta_j-\theta_k)}} } = - \pi - \si\abs{\frac{2 x(\theta_j) x(\theta_k)}{\sin(\theta_j-\theta_k)}}.
\end{equation}
\begin{widetext}
Therefore, the Wigner function can be simplified to
\begin{equation}
\begin{aligned}
    W_{A(\vec{\theta})}(x,p) &= \frac{1}{2\pi K^2}\sum_{j=0}^{K-1}\sum_{k=0}^{K-1} \sgn[x(\theta_j)x(\theta_k)] \si\abs{\frac{2 x(\theta_j) x(\theta_k)}{\sin(\theta_j-\theta_k)} } \\
    W_{A(\vec{\theta})}(r\cos\phi,r\sin\phi) &= \frac{1}{2\pi K^2}\sum_{j=0}^{K-1}\sum_{k=0}^{K-1} \sgn[\cos(\phi-\theta_j)\cos(\phi-\theta_k)] \si\abs{\frac{2 r^2 \cos(\phi-\theta_j) \cos(\phi-\theta_k)}{\sin(\theta_j-\theta_k)} }.
\end{aligned}
\end{equation}
Hence, ${\tr}[A(\vec{\theta})^2]$ is
\begin{equation}\label{eq:upper-bound-general-angle-integral}
\begin{aligned}
    \tr[A(\vec{\theta})^2] &= \frac{1}{16\pi^3 K^4} \int_{-\pi}^\pi\dd{\phi} \int_{0}^\infty\dd{[r^2]} \sum_{j=0}^{K-1}\sum_{k=0}^{K-1}\sum_{l=0}^{K-1}\sum_{m=0}^{K-1}
    \sgn[\cos(\phi-\theta_j)\cos(\phi-\theta_k)\cos(\phi-\theta_l)\cos(\phi-\theta_m)] \\
    &\hspace{20em}{}\times{}\si\abs{\frac{2 r^2 \cos(\phi-\theta_j) \cos(\phi-\theta_k)}{\sin(\theta_j-\theta_k)} }\si\abs{\frac{2 r^2 \cos(\phi-\theta_l) \cos(\phi-\theta_m)}{\sin(\theta_l-\theta_m)} } \\
    &= \frac{1}{64\pi^2 K^4} \int_{-\pi}^\pi\dd{\phi}\sum_{j=0}^{K-1}\sum_{k=0}^{K-1}\sum_{l=0}^{K-1}\sum_{m=0}^{K-1} \sgn[\cos(\phi-\theta_j)\cos(\phi-\theta_k)\cos(\phi-\theta_l)\cos(\phi-\theta_m)] \\
    &\hspace{20em}{}\times{}\min\Bqty{
    \abs{\frac{\sin(\theta_j-\theta_k)}{\cos(\phi-\theta_j) \cos(\phi-\theta_k)} },\abs{\frac{\sin(\theta_l-\theta_m)}{\cos(\phi-\theta_l) \cos(\phi-\theta_m)} }
    }.
\end{aligned}
\end{equation}
\end{widetext}
Now, the remaining integral involves an integrand that is the sum of piecewise smooth functions: apart from nondifferentiable points whenever the sign function flips signs or the minimum function swaps between its arguments, the function is simply the reciprocal of two cosines, which is smooth within the nondifferentiable points.

We can therefore split the integral as $-\pi \leq \phi \leq \pi = \sum_{n=0}^{|\Phi|+1} \int_{\phi_{n}}^{{\phi}_{n+1}}$, where $\phi_0 =-\pi$, $\phi_{|\Phi|+1} = \pi$, and $\phi_n \in \Phi$ is the $n$th smallest element of the set
\begin{equation}
\begin{aligned}
    \Phi = \Bigg\{
    \phi {}:{} & |\phi| \leq \pi : \exists j,k,l,m : \cos(\phi-\theta_k) = 0 \\[-1ex]
    &{}\lor{} \pm \frac{|\sin(\theta_j-\theta_k)|}{\cos(\phi-\theta_j)\cos(\phi-\theta_k)} \\
    &\qquad{}={} \frac{|\sin(\theta_l-\theta_m)|}{\cos(\phi-\theta_l)\cos(\phi-\theta_m)}
    \Bigg\}.
\end{aligned}
\end{equation}
Here, $\Phi$ includes all points $\phi$ at which the $\sgn$ and $\min$ functions are nondifferentiable. Note that closed-form expressions of the elements of $\Phi$ can be found---the condition $\cos(\phi-\theta_k)=0 \implies \phi = \theta_k + (2l+1)\pi/2$ for integer $l$, and the condition involving reciprocals can be rewritten into a quadratic polynomial of $\tan\phi$. This means that we can rewrite Eq.~\eqref{eq:upper-bound-general-angle-integral} as a sum of integrals over smooth functions.

Within each integral, after taking care to use trigonometric identities to regulate the terms whose denominators go to zero in the same way as Eq.~\eqref{eq:upper-bound-angle-integral}, each integral can be analytically evaluated to be
\begin{equation}
\begin{aligned}
    &\int_{\phi_n}^{\phi_{n+1}}\dd{\phi} \frac{\pm \sin(\theta_j-\theta_k)}{\cos(\phi-\theta_j)\cos(\phi-\theta_k)} \\
    &\qquad{}={} \pm\log\!\bqty{\frac{{\cos}(\phi_{n+1}-\theta_j){\cos}(\phi_{n}-\theta_k)}{{\cos}(\phi_{n+1}-\theta_k){\cos}(\phi_{n}-\theta_j)}}.
\end{aligned}
\end{equation}
This means that, in principle, closed-form solutions of Eq.~\eqref{eq:upper-bound-general-angle-integral} can be evaluated analytically. In practice, however, the obtained expressions become very cumbersome for $K > 3$. Nonetheless, this shows that the value obtained from numerical integration of Eq.~\eqref{eq:upper-bound-general-angle-integral} is reliable.

Finally, the expression ${\tr}[A(\vec{\theta})^2] \geq 2[(\mathbf{P}^\infty(\vec{\theta})-\tfrac{1}{2})^2-(1/2K)^2]^2$ can be rearranged to find
\begin{equation}
    \mathbf{P}^{\geq}(\vec{\theta}) := \frac{1}{2}\pqty{1 + \frac{1}{K}\sqrt{1 + 2K^2\sqrt{2{\tr}[A(\vec{\theta})^2]}}} \geq \mathbf{P}^{\infty}(\vec{\theta}).
\end{equation}

\section{\label{apd:spin-argument-convergence}Location of Local Maxima for Precession Protocol With More Angles}
In this section, we compare the location of the local maxima of the precession protocol with $K$ angles to the resonant angles
$\vec{\theta}^{(j)}_{\Delta} = ( \pi n_k/j )_{k=1}^{K-1}$ defined analogously to the $K=3$ case, where $\{n_k\}_{k=1}^{K-1}$ are integers such that
\begin{equation}
0 < n_1 < \dots < n_{\frac{K-1}{2}} \leq \lfloor j \rfloor < n_{\frac{K+1}{2}} < \dots \leq K-1.
\end{equation}
In Fig.~\ref{fig:K5-parameter}, we plot the distance between $\vec{\theta}$ and the closest mixture $\vec{\theta}_{\lambda}^{(j)} = (1-\lambda)\vec{\theta}^{(j)}_{\Delta} + \lambda \vec{\theta}_K$ as we increase the number of iterations of the gradient descent algorithm for $j \in \{7/2,11/2\}$. Many initial points converge to $\vec{\theta}_{\lambda}^{(j)}$, which shows that they are still local maxima even in the $K=5$ case. However, we also observe local maxima that are not of the form $(1-\lambda)\vec{\theta}^{(j)}_{\Delta} + \lambda \vec{\theta}_K$.

\begin{figure}
    \centering
    \includegraphics[width=\columnwidth]{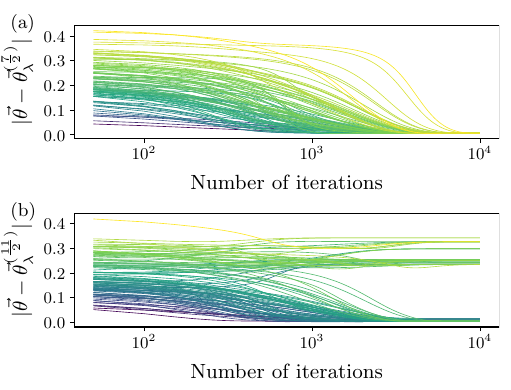}
    \caption{\label{fig:K5-parameter}\textbf{Distance of probing angles $\vec{\theta}$ to closest mixture $\vec{\theta}_\lambda^{(j)} =  (1-\lambda)\vec{\theta}_{\Delta}^{(j)} + \lambda\vec{\theta}_K$ against the number of iterations of the gradient descent, with $K=5$.} This is done for (a) $j=7/2$ and (b) $j=11/2$. The location of many local peaks are close to $\vec{\theta}_\lambda^{(j)}$, although there are some local peaks that are not.}
\end{figure}

\begin{figure}
    \centering
    \includegraphics[width=\columnwidth]{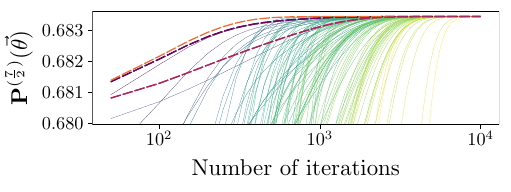}
    \caption{\label{fig:K5-initial-point}\textbf{Convergence to local maximum with the number of iterations.} By choosing the initial parameters $(\vec{\theta}^{(j)}_{\Delta} + \vec{\theta}_K)/2$ for the gradient descent, shown with bolder lines, large violations can be found much faster than choosing random initial parameters.}
\end{figure}

That said, this observation helps us choose good initial points for performing the gradient descent. In Fig.~\ref{fig:K5-initial-point}, we contrast random initial points with initial points of the form $(\vec{\theta}^{(j)}_{\Delta} + \vec{\theta}_K)/2$. We find that the chosen initial points already start with large violations, and converge much faster to the local maximum than randomly chosen parameters. This aids in finding states and parameters $\vec{\theta}$ that obtain large scores $\mathbf{P}^{(j)}(\vec{\theta})$, similar to those prepared in recent experimental implementations of the precession protocol \cite{tsirelson-experiment}.


\end{document}